\documentclass[preprint,12pt]{elsarticle} 

\usepackage[margin=1in]{geometry}

\usepackage{amsmath}   
\usepackage{amssymb}   
\usepackage{amsthm}    
\usepackage{bm}        
\usepackage{mathtools} 

\usepackage{graphicx}     
\usepackage{booktabs}     
\usepackage{multirow}     
\usepackage{algorithm}     
\usepackage{algpseudocode} 

\usepackage{float}         
\floatstyle{boxed}
\restylefloat{algorithm}

\usepackage{xcolor}
\usepackage[colorlinks=true,
            linkcolor=blue,
            citecolor=blue,
            urlcolor=blue]{hyperref}

\theoremstyle{plain}
\newtheorem{theorem}{Theorem}[section]

\theoremstyle{definition}

\theoremstyle{remark}
\newtheorem{remark}{Remark}[section]

\begin{document}

\title{Joint Discovery of Graph Structure and Dynamics in Stochastic Interacting Particle Systems}

\makeatletter
\def\blfootnote#1{%
  \begingroup
  \renewcommand\thefootnote{}\footnotetext{#1}%
  \addtocounter{footnote}{-1}%
  \endgroup
}
\renewcommand{\footnoterule}{%
  \kern -3pt
  \hrule width 0.32\textwidth
  \kern 2.6pt
}
\makeatother

 \author{
 \normalsize{
Demao Liu$^{1,2}$,
Ting Gao$^{1,2*}$,
Jinqiao Duan$^{3,4}$
}\\[10pt]
\footnotesize{$^1$ School of Mathematics and Statistics, Huazhong University of Science and Technology, China} \\
\footnotesize{$^2$ Center for Mathematical Science, Huazhong University of Science and Technology, China} \\
\footnotesize{$^3$ Department of Mathematics and Department of Physics, Great Bay University, China} \\
\footnotesize{$^4$ Guangdong Provincial Key Laboratory of Mathematical and Neural dynamic Systems, China. }
}

\begin{abstract}

We study the joint identification of network structure and governing dynamics in stochastic interacting particle systems, which consist of an unknown directed weighted interaction graph with unknown local and non-local interaction components. We formulate the problem as a coupled inverse problem for the graph and the associated basis coefficients, and develop two alternating least-squares-type estimators: a three-block scheme (TALS) and an integrated diagonal-augmented scheme (IALS). The IALS formulation combines the updates of the local and interaction coefficients into a single least-squares subproblem, and is particularly well suited to settings in which the nodewise local dynamics share a common functional template up to node-dependent scaling. We further establish an identifiability result under a rank-2 joint coercivity
condition together with an appropriate normalization convention. Synthetic experiments show that the proposed estimators accurately recover both the interaction graph and the dynamical components, and remain robust under stochastic forcing, observation noise, and basis mismatch. We also provide an illustrative real-data application on ictal
SEEG recordings, where the learned models produce stable and interpretable dynamical summaries across multiple basis configurations. This work advances a theoretically guaranteed scalable framework for learning stochastic interacting particle systems, with broad potential for data-driven identification in computational biology, neuroscience, and beyond.

\end{abstract}


\maketitle  
\blfootnote{* Corresponding Author, \texttt{tgao0716@hust.edu.cn}}

\section{Introduction}

Recent advances in large-scale data acquisition have made it possible to observe high-dimensional interacting systems in a wide range of applications, including neuronal population dynamics, collective biological systems, and other networked stochastic processes\cite{wang2024virtual,urai2022large,vicsek2012collective}. In many such settings, the available data consist of time-resolved trajectories from multiple agents evolving over an unknown interaction network. A central problem is therefore to infer from data both the governing dynamics and the underlying interaction structure\cite{spohn2012large,bongini2017inferring}.

A broad class of heterogeneous interacting systems on a directed weighted graph can be modeled as
\begin{equation}
dX_t^i
=
\Bigl(
L(X_t^i)
+
\sum_{j\in [N]\setminus\{i\}} a_{ij}\,\Phi(X_t^j-X_t^i)
\Bigr)\,dt
+
\sigma\, dW_t^i,
\qquad i\in[N],
\label{eq:model}
\end{equation}
where $X_t^i\in\mathbb{R}^d$ is the state of agent $i$, $L:\mathbb{R}^d\to\mathbb{R}^d$ is the local drift field, $\Phi:\mathbb{R}^d\to\mathbb{R}^d$ is the interaction kernel, $A=(a_{ij})$ is a directed weighted adjacency matrix, and $\{W_t^i\}_{i=1}^N$ are independent Brownian motions. This formulation provides a unified framework in which intrinsic dynamics, network-mediated interactions, and stochastic forcing are explicitly separated.

In much of the existing literature, one or more of the interaction matrix $A$, the local drift $L$, and the interaction kernel $\Phi$ are assumed to be known, prescribed up to a restricted parametric form\cite{gao2022autonomous,gao2024learning}. In practical applications, however, neither the interaction matrix $A$ nor the functional forms of $L$ and $\Phi$ are typically known in advance. Jointly recovering them from discrete and noisy observations gives rise to a high-dimensional nonlinear inverse problem \cite{lu2022learning}. The difficulty is further compounded by the coupling between interaction strengths and interaction functions: without suitable structural constraints, this coupling may lead to scale non-identifiability and ill-conditioned estimation.

A common strategy in recent data-driven modeling is to represent unknown nonlinearities by highly expressive black-box approximators\cite{qin2019data}. While such approaches can be effective for prediction, they often obscure the mechanistic content of the recovered dynamics and limit interpretability in applications where prior structural information is available\cite{chen2021generalized}. In many biological and physical systems, the local drift and interaction response are not arbitrary nonlinear maps; instead, they often exhibit structured features such as thresholding, saturation, or other low-dimensional response patterns\cite{wilson1972excitatory,ferrell2014ultrasensitivity}. This motivates the use of structured functional parameterizations that incorporate prior knowledge while retaining sufficient expressive flexibility.

In this work, we represent both the local field $L$ and the interaction kernel $\Phi$ by linear expansions over preselected basis functions. This yields a computationally tractable model class and allows biologically or physically meaningful priors to be encoded explicitly in the representation. Importantly, the interaction matrix $A$ is not prescribed beforehand, but is inferred jointly with the basis coefficients from observed trajectories. The resulting formulation performs simultaneous recovery of the network topology and the dynamical components within a unified stochastic interacting-system framework.

The problem therefore lies at the intersection of system identification, stochastic dynamical inference, and network reconstruction. Our objective is to develop an interpretable and data-adaptive framework for estimating the weighted interaction graph together with the local and interaction components directly from discrete-time trajectory data.

\section{Problem Statement}

\subsection{Unknown parameters and parameterized model}

We consider the joint estimation of the interaction matrix $A$ and the basis coefficients that determine the unknown local drift $L$ and interaction kernel $\Phi$ in \eqref{eq:model}. Specifically, the unknowns are the weighted adjacency matrix $A=(a_{ij})_{i,j=1}^N\in[0,1]^{N\times N}$ with $a_{ii}=0$, together with the coefficient vectors $b=(b_1,\dots,b_q)\in\mathbb{R}^q$ and $c=(c_1,\dots,c_p)\in\mathbb{R}^p$. We adopt the normalization convention $\sum_{j\in[N]\setminus\{i\}} a_{ij}^2=1$ for each $i\in[N]$. The coefficients $b$ and $c$ determine the local drift and interaction kernel through the basis expansions
\[
L(x)=\sum_{\ell=1}^q b_\ell\,\phi_\ell(x),
\qquad
\Phi(x)=\sum_{k=1}^p c_k\,\psi_k(x).
\]

Substituting these expansions into \eqref{eq:model} gives the parameterized system
\begin{equation}\label{eq:param_model}
dX_t^i
=
\left(
\sum_{\ell=1}^q b_\ell\,\phi_\ell(X_t^i)
+
\sum_{j\in[N]\setminus\{i\}} a_{ij}
\sum_{k=1}^p c_k\,\psi_k(X_t^j-X_t^i)
\right)\,dt
+
\sigma\, dW_t^i,
\qquad i\in[N].
\end{equation}
This representation makes the unknown parameters explicit and provides the basis for the estimation problem considered below.

\subsection{Observation setting and estimation objective}

The available data consist of $M$ independent sample trajectories observed on a uniform time grid, denoted by $\{\mathbf{X}^{(m)}_{t_0:t_L}\}_{m=1}^M$, where $\mathbf{X}^{(m)}_{t_\ell}\in\mathbb{R}^{N\times d}$ is the full system state of the $m$-th trajectory at time $t_\ell$, with $t_\ell=t_0+\ell\Delta t$ for $\ell=0,1,\dots,L$. The observation horizon is therefore $T=L\Delta t$, and each initial state is sampled as $\mathbf{X}^{(m)}_{t_0}\sim\mu^{\otimes N}$ for some probability distribution $\mu$ on $\mathbb{R}^d$.

Given only these discrete-time trajectories, our objective is to estimate the parameter triple $\Theta=(A,b,c)$, and thereby recover simultaneously the directed weighted interaction structure, the local drift $L$, and the interaction kernel $\Phi$.

\subsection{Compact notation}

To write the drift term compactly for all agents simultaneously, we introduce the following notation. Here $\times_n$ denotes the standard mode-$n$ tensor--vector contraction. For the local drift term, define the feature tensor $\mathcal{H}(\mathbf{X}_t)\in\mathbb{R}^{N\times d\times q}$ by
\[
\bigl(\mathcal{H}(\mathbf{X}_t)\bigr)_{i,:, \ell}=\phi_\ell(X_t^i)\in\mathbb{R}^d,
\qquad i\in[N],\ \ell\in[q].
\]
For the interaction term, define the interaction feature tensor $\mathcal{B}(\mathbf{X}_t)\in\mathbb{R}^{N\times N\times d\times p}$ by
\[
\bigl(\mathcal{B}(\mathbf{X}_t)\bigr)_{i,j,:,k}
=
\begin{cases}
\psi_k(X_t^j-X_t^i), & j\neq i,\\
0, & j=i,
\end{cases}
\qquad i,j\in[N],\ k\in[p].
\]

Using these quantities, we introduce the aggregated local and interaction terms
\[
\mathcal{L}(\mathbf{X}_t;b):=\mathcal{H}(\mathbf{X}_t)\times_3 b \in \mathbb{R}^{N\times d},
\qquad
\mathcal{U}(\mathbf{X}_t;A,c)\in\mathbb{R}^{N\times d},
\]
where the $i$-th rows are given by
\[
\mathcal{L}(\mathbf{X}_t;b)_{i,:}=\sum_{\ell=1}^{q} b_\ell\,\phi_\ell(X_t^i),
\qquad
\mathcal{U}(\mathbf{X}_t;A,c)_{i,:}
=
\sum_{j\in[N]} a_{ij}\,(\mathcal{B}(\mathbf{X}_t)\times_4 c)_{i,j,:},
\]
with
\[
(\mathcal{B}(\mathbf{X}_t)\times_4 c)_{i,j,:}
=
\sum_{k=1}^{p} c_k\,\bigl(\mathcal{B}(\mathbf{X}_t)\bigr)_{i,j,:,k}.
\]

Substituting the basis expansions into \eqref{eq:param_model}, we obtain the compact system form
\begin{equation}\label{eq:tensor_form}
dX_t^i
=
\Bigl(
\mathcal{L}(\mathbf{X}_t;b)_{i,:}
+
\mathcal{U}(\mathbf{X}_t;A,c)_{i,:}
\Bigr)\,dt
+
\sigma\, dW_t^i,
\qquad i\in[N].
\end{equation}
Equivalently, in array form,
\begin{equation}\label{eq:tensor_array_form}
d\mathbf{X}_t
=
\Bigl(
\mathcal{L}(\mathbf{X}_t;b)
+
\mathcal{U}(\mathbf{X}_t;A,c)
\Bigr)\,dt
+
\sigma\, d\mathbf{W}_t,
\end{equation}
where $\mathbf{X}_t,\mathbf{W}_t\in\mathbb{R}^{N\times d}$ collect the states and Brownian motions of all agents, respectively.

\subsection{Empirical loss and constrained estimation problem}

Given $M$ trajectories $\{\mathbf{X}^{(m)}_{t_0:t_L}\}_{m=1}^M$ observed on a uniform grid $t_\ell=t_0+\ell\Delta t$, define the increment
\[
\Delta \mathbf{X}^{(m)}_{t_\ell}
:=
\mathbf{X}^{(m)}_{t_{\ell+1}}-\mathbf{X}^{(m)}_{t_\ell}
\in\mathbb{R}^{N\times d},
\qquad \ell=0,\dots,L-1.
\]

Since the system is observed only at discrete time points, we use the Euler--Maruyama discretization of \eqref{eq:tensor_array_form} to obtain a discrete regression relation for the drift term. This leads to the empirical loss
\begin{equation}\label{eq:loss_function}
\mathcal{E}_{L,M}(A,b,c)
:=
\frac{1}{ML}\sum_{m=1}^{M}\sum_{\ell=0}^{L-1}
\left\|
\frac{\Delta \mathbf{X}^{(m)}_{t_\ell}}{\Delta t}
-
\mathcal{L}(\mathbf{X}^{(m)}_{t_\ell};b)
-
\mathcal{U}(\mathbf{X}^{(m)}_{t_\ell};A,c)
\right\|_{F}^{2},
\end{equation}
where $\|\cdot\|_{F}$ denotes the Frobenius norm on $\mathbb{R}^{N\times d}$.

We then estimate $(A,b,c)$ via empirical risk minimization:
\begin{equation*}
(\widehat{A},\widehat{b},\widehat{c})
\in
\arg\min_{(A,b,c)\in \mathcal{M}\times\mathbb{R}^q\times\mathbb{R}^p}
\mathcal{E}_{L,M}(A,b,c),
\end{equation*}
where $\mathcal{M}$ is the admissible set for $A$, for example,
\[
\mathcal{M}
=
\left\{
A\in[0,1]^{N\times N}:\ 
a_{ii}=0,\ 
\sum_{j\in[N]\setminus\{i\}} a_{ij}^{2}=1,\ \forall i\in[N]
\right\}.
\]

If $\sigma=0$ and the time derivatives $\dot{\mathbf{X}}^{(m)}_{t_\ell}$ are directly available, the finite-difference quotient $\Delta \mathbf{X}^{(m)}_{t_\ell}/\Delta t$ in \eqref{eq:loss_function} may be replaced by $\dot{\mathbf{X}}^{(m)}_{t_\ell}$. When $\sigma>0$, the Euler--Maruyama increments are approximately Gaussian, with covariance of order $\sigma^2\Delta t$, which provides a likelihood-based justification for \eqref{eq:loss_function} as a least-squares criterion up to an additive constant.

The minimization of \eqref{eq:loss_function} with respect to $(A,b,c)$ is a coupled nonconvex optimization problem. In particular, the interaction matrix $A$ and the coefficient vectors $b$ and $c$ enter the loss in a nonlinear and mutually dependent manner, which makes direct joint optimization computationally challenging. By contrast, when two blocks of variables are fixed, the remaining subproblem becomes substantially simpler; in particular, the updates for the basis coefficients reduce to least-squares or, more generally, convex quadratic subproblems, while the update for $A$ can also be treated within a structured constrained optimization framework. This observation motivates the use of an alternating optimization strategy, in which the variables are updated block by block so as to exploit the favorable structure of the corresponding subproblems.

\section{Methods}

\subsection{Three-block Alternating Least-squares scheme(TALS)}

To minimize the empirical loss \eqref{eq:loss_function}, we employ a three-block alternating least-squares scheme, referred to as TALS for brevity. The variables are updated cyclically in the order $A \rightarrow c \rightarrow b$, while the remaining two blocks are kept fixed. This strategy exploits the least-squares or quadratic structure of the resulting subproblems and avoids direct joint optimization over the coupled nonconvex objective.

\medskip
\noindent{\normalsize\bfseries{Update of the interaction matrix $A$}.}
Assume that the current iterates $b$ and $c$ are fixed. The update of $A$ is performed row by row. For each $i\in[N]$, define
\[
y_i^{(m,\ell)}(b)
:=
\frac{(\Delta \mathbf{X}_{t_\ell}^{(m)})_{i,:}}{\Delta t}
-
\mathcal{L}(\mathbf{X}_{t_\ell}^{(m)};b)_{i,:}
\in \mathbb{R}^d
\]
as the residual after removing the local contribution, and define
\[
z_{ij}^{(m,\ell)}(c)
:=
\bigl(\mathcal{B}(\mathbf{X}_{t_\ell}^{(m)})\times_4 c\bigr)_{i,j,:}
\in \mathbb{R}^d,
\qquad j\neq i,
\]
as the interaction feature associated with source node $j$. Stacking these quantities over all trajectories and time steps yields a response vector $y_i(b)\in\mathbb{R}^{dLM}$ and a design matrix $Z_i(c)\in\mathbb{R}^{dLM\times (N-1)}$.

Let $a_{i,-i}\in\mathbb{R}^{N-1}$ denote the vector of off-diagonal entries in the $i$-th row of $A$. The $i$-th row is then updated by solving the nonnegative least-squares problem
\begin{equation}\label{eq:A_nnls}
a_{i,-i}^{\mathrm{NNLS}}
\in
\arg\min_{u\in\mathbb{R}_{+}^{N-1}}
\|y_i(b)-Z_i(c)u\|_2^2,
\end{equation}
where the diagonal entry $a_{ii}$ is excluded from the optimization and kept equal to zero. After obtaining the provisional solution $\widetilde a_{i,-i}$, we perform a row-wise normalization projection,
\begin{equation}\label{eq:A_projection}
a_{i,-i}
=
\Pi_{\mathbb{S}_{+}^{N-2}}(a_{i,-i}^{\mathrm{NNLS}})
:=
\frac{a_{i,-i}^{\mathrm{NNLS}}}{\|a_{i,-i}^{\mathrm{NNLS}}\|_2},
\qquad \text{provided that } a_{i,-i}^{\mathrm{NNLS}}\neq 0,
\end{equation}
where $\mathbb{S}_{+}^{N-2}:=\{u\in\mathbb{R}_{+}^{N-1}:\|u\|_2=1\}$. The diagonal entry is then set to $a_{ii}=0$. In this way, the update preserves nonnegativity and enforces the row-wise normalization constraint.

\medskip
\noindent{\normalsize\bfseries{Update of the interaction coefficient vector $c$}.}
With $A$ and $b$ fixed, the loss \eqref{eq:loss_function} becomes a least-squares problem in $c$. Let
\[
Y_c(A,b)
:=
\left[
\frac{\Delta \mathbf{X}_{t_\ell}^{(m)}}{\Delta t}
-
\mathcal{L}(\mathbf{X}_{t_\ell}^{(m)};b)
\right]_{m,\ell}
\]
denote the stacked residual after removing the local term, and let
\[
Z_c(A)
:=
\bigl[\mathcal{U}(\mathbf{X}_{t_\ell}^{(m)};A,\cdot)\bigr]_{m,\ell}
\]
be the corresponding design matrix acting on $c$. Then $c$ is updated by solving
\begin{equation}\label{eq:c_update}
c
\in
\arg\min_{\xi\in\mathbb{R}^{p}}
\|Y_c(A,b)-Z_c(A)\xi\|_2^2.
\end{equation}

\medskip
\noindent{\normalsize\bfseries{Update of the local coefficient vector $b$}.}
With $A$ and $c$ fixed, the loss \eqref{eq:loss_function} becomes a least-squares problem in $b$. Let
\[
Y_b(A,c)
:=
\left[
\frac{\Delta \mathbf{X}_{t_\ell}^{(m)}}{\Delta t}
-
\mathcal{U}(\mathbf{X}_{t_\ell}^{(m)};A,c)
\right]_{m,\ell}
\]
denote the stacked residual after removing the interaction term, and let
\[
Z_b
:=
\bigl[\mathcal{L}(\mathbf{X}_{t_\ell}^{(m)};\cdot)\bigr]_{m,\ell}
\]
be the corresponding design matrix acting on $b$. Then $b$ is updated by solving
\begin{equation}\label{eq:b_update}
b
\in
\arg\min_{\eta\in\mathbb{R}^{q}}
\|Y_b(A,c)-Z_b\eta\|_2^2.
\end{equation}

\medskip
\noindent{\normalsize\bfseries{Iteration.}}
Starting from an initial guess $(A^{(0)},b^{(0)},c^{(0)})$, we repeat the above three updates in Gauss--Seidel fashion until convergence, namely until the relative changes of all three blocks fall below a prescribed tolerance $\epsilon$, or until a maximum number of iterations $n_{\max}$ is reached.

\begin{algorithm}[htbp]
\caption{TALS: three-block alternating least-squares scheme}
\label{alg:tals}
\begin{algorithmic}
\Procedure{TALS}{$\{\mathbf{X}^{(m)}_{t_0:t_L}\}_{m=1}^M,\ \epsilon,\ n_{\max}$}
    \State Construct the empirical quantities in \eqref{eq:loss_function}.
    \State Initialize $A^{(0)},\, b^{(0)},\, c^{(0)}$ with $A^{(0)}_{ii}=0$ for all $i$.
    \For{$\tau=1,\dots,n_{\max}$}
        \State Update $A^{(\tau)}$ rowwise by solving \eqref{eq:A_nnls} and applying the normalization \eqref{eq:A_projection}, with $b=b^{(\tau-1)}$ and $c=c^{(\tau-1)}$ fixed.
        \State Update $c^{(\tau)}$ by solving \eqref{eq:c_update}, with $A=A^{(\tau)}$ and $b=b^{(\tau-1)}$ fixed.
        \State Update $b^{(\tau)}$ by solving \eqref{eq:b_update}, with $A=A^{(\tau)}$ and $c=c^{(\tau)}$ fixed.
        \If{the relative changes of $A$, $c$, and $b$ are all below $\epsilon$}
            \State \textbf{break}
        \EndIf
    \EndFor
    \State \Return $A^{(\tau)},\, b^{(\tau)},\, c^{(\tau)}$
\EndProcedure
\end{algorithmic}
\end{algorithm}

\subsection{Integrated Alternating Least-squares Scheme (IALS)}

Compared with TALS, we also consider a reduced two-block scheme in which the local and interaction coefficients are updated together. In the original problem formulation, the interaction matrix $A$ is defined with zero diagonal, and only its off-diagonal entries represent interaction weights. In IALS, we retain this interpretation for the off-diagonal interaction part, but introduce a diagonal-augmented coefficient matrix $\widetilde A=(\widetilde a_{ij})\in\mathbb{R}_+^{N\times N}$ in order to absorb the local term into the same linear-weight structure. This allows the local and interaction contributions to be treated within a unified least-squares framework.

The key observation is that, under an additional shared-template assumption on the local dynamics, the nodewise local term can be written as a node-dependent scaling of a common local template. More precisely, we assume that
\[
L_i(x)=\alpha_i L_0(x),
\qquad
L_0(x):=\sum_{\ell=1}^{q}\beta_\ell\,\phi_\ell(x),
\]
for some nonnegative scaling factors $\alpha_i$. Under this reparameterization, the local scaling factors are encoded by the diagonal entries of $\widetilde A$, while the off-diagonal entries continue to encode the interaction weights. This leads to the following diagonal-augmented parameterization:
\begin{equation}\label{eq:ials_model}
dX_t^i
=
\left(
\widetilde a_{ii}\sum_{\ell=1}^{q}\beta_\ell\,\phi_\ell(X_t^i)
+
\sum_{j\in[N]\setminus\{i\}}
\widetilde a_{ij}\sum_{k=1}^{p}c_k\,\psi_k(X_t^j-X_t^i)
\right)\,dt
+
\sigma\,dW_t^i,
\qquad i\in[N].
\end{equation}

Here the off-diagonal entries $\widetilde a_{ij}$, $j\neq i$, have the same interpretation as in the original formulation, namely as interaction weights. By contrast, the diagonal entries $\widetilde a_{ii}$ are not interpreted as self-loop weights. Instead, they are node-dependent scaling factors associated with the shared local template $L_0$. Accordingly, \eqref{eq:ials_model} is exact whenever the local dynamics admit the form $L_i(x)=\alpha_iL_0(x)$ for some nonnegative scaling factors $\alpha_i$. In particular, when all nodes share the same local dynamics, the learned diagonal factors are expected to be approximately constant across nodes, up to the normalization convention adopted in Algorithm~\ref{alg:ials}.

To combine the interaction and local terms into a unified least-squares structure, define the augmented coefficient vector $\theta:=(c^\top,\beta^\top)^\top\in\mathbb{R}^{p+q}$, where $c=(c_1,\dots,c_p)^\top$ and $\beta=(\beta_1,\dots,\beta_q)^\top$. For each target node $i$, trajectory $m$, and time step $\ell$, define the augmented feature tensor $\widetilde{\mathcal B}_i(\mathbf X_{t_\ell}^{(m)})\in\mathbb R^{N\times d\times (p+q)}$ by
\[
\bigl(\widetilde{\mathcal B}_i(\mathbf X_{t_\ell}^{(m)})\bigr)_{j,:,s}
=
\begin{cases}
\psi_s(X_{t_\ell}^{j,(m)}), & j\neq i,\ s\in[p],\\[1mm]
\phi_{s-p}(X_{t_\ell}^{i,(m)}), & j=i,\ s\in\{p+1,\dots,p+q\},\\[1mm]
0, & \text{otherwise}.
\end{cases}
\]
Thus, for a fixed target node $i$, the $i$-th slice in the node dimension contains the local basis block, while the remaining slices contain the interaction basis block. With this notation, the drift term in \eqref{eq:ials_model} can be written compactly as
\[
\widetilde{\mathcal B}_i(\mathbf X_{t_\ell}^{(m)})\times_1 \widetilde a_{i\cdot}\times_3 \theta\in\mathbb R^d,
\]
where $\widetilde a_{i\cdot}\in\mathbb R^N$ denotes the $i$-th row of $\widetilde A$.

\medskip
\noindent{\normalsize\bfseries Update of the diagonal-augmented coefficient matrix $\widetilde A$.}
For fixed $\theta$, define the response vector for node $i$ by
\[
\widetilde y_i:=\bigl[(\Delta \mathbf X_{t_\ell}^{(m)})_{i,:}/\Delta t\bigr]_{m,\ell}\in\mathbb R^{dLM},
\]
and define the corresponding nodewise design matrix $\widetilde Z_i(\theta)\in\mathbb R^{dLM\times N}$ by stacking the sample matrices
\[
\bigl(\widetilde{\mathcal B}_i(\mathbf X_{t_\ell}^{(m)})\times_3 \theta\bigr)^\top\in\mathbb R^{d\times N}
\]
over $(m,\ell)$. The rows of $\widetilde A$ are then updated independently by solving the nonnegative least-squares problem
\begin{equation}\label{eq:ials_A_update}
\widetilde a_{i\cdot}
\in
\arg\min_{u\in\mathbb R_+^{N}}
\|\widetilde y_i-\widetilde Z_i(\theta)u\|_2^2,
\qquad i\in[N].
\end{equation}

After each NNLS update, we normalize the off-diagonal interaction block and the diagonal local-scaling vector separately. More precisely, for each $i\in[N]$, let
\[
\widetilde a_{i,-i}:=(\widetilde a_{ij})_{j\neq i}\in\mathbb R_+^{N-1},
\qquad
d(\widetilde A):=(\widetilde a_{11},\dots,\widetilde a_{NN})^\top\in\mathbb R_+^{N}.
\]
We then apply the projections
\[
\widetilde a_{i,-i}\leftarrow \widetilde a_{i,-i}/\|\widetilde a_{i,-i}\|_2
\qquad \text{whenever } \widetilde a_{i,-i}\neq 0,
\]
for all $i\in[N]$, and
\[
d(\widetilde A)\leftarrow d(\widetilde A)/\|d(\widetilde A)\|_2
\qquad \text{whenever } d(\widetilde A)\neq 0.
\]
Thus, the interaction rows are normalized independently across nodes, while the diagonal entries are normalized collectively as a separate local-scaling vector. This split normalization is introduced to reduce the intrinsic scale ambiguity between $\widetilde A$ and $\theta$ while preserving the distinct roles of the off-diagonal interaction weights and the diagonal local coefficients. In particular, it makes the off-diagonal block directly comparable to the interaction graph, whereas the diagonal block is interpreted only through the reconstructed local term. Its role in the identifiability discussion is clarified in Remark~\ref{rem:ials_normalization} and Theorem~\ref{thm:ials_identifiability} below.

\medskip
\noindent{\normalsize\bfseries Update of the augmented coefficient vector $\theta$.}
For fixed $\widetilde A$, let $\widetilde Y\in\mathbb R^{dNLM}$ denote the global response vector obtained by stacking all nodewise responses, and let $\widetilde W(\widetilde A)\in\mathbb R^{dNLM\times (p+q)}$ denote the global design matrix obtained by stacking the sample matrices
\[
\widetilde{\mathcal B}_i(\mathbf X_{t_\ell}^{(m)})\times_1 \widetilde a_{i\cdot}\in\mathbb R^{d\times (p+q)}
\]
over $(i,m,\ell)$. The empirical loss then becomes a least-squares problem in $\theta$:
\begin{equation}\label{eq:ials_theta_update}
\theta
\in
\arg\min_{\xi\in\mathbb R^{p+q}}
\|\widetilde Y-\widetilde W(\widetilde A)\xi\|_2^2.
\end{equation}
After solving \eqref{eq:ials_theta_update}, the first $p$ entries of $\theta$ are identified with the interaction coefficients $c$, and the remaining $q$ entries are identified with the local-template coefficients $\beta$.

\medskip
\noindent{\normalsize\bfseries Iteration.}
Starting from an initial guess $(\widetilde A^{(0)},\theta^{(0)})$, we alternate between the NNLS update \eqref{eq:ials_A_update} and the least-squares update \eqref{eq:ials_theta_update}, followed by the split normalization of the off-diagonal interaction block and the diagonal vector, until convergence. Compared with TALS, this integrated scheme reduces the number of coefficient subproblems solved per iteration from two to one. The trade-off is that the method relies on the diagonal-augmented parameterization \eqref{eq:ials_model}, and is therefore most suitable when the local dynamics across nodes share a common functional form up to node-dependent scaling.

\begin{algorithm}[htbp]
\caption{IALS: integrated alternating least-squares scheme}
\label{alg:ials}
\begin{algorithmic}
\Procedure{IALS}{$\{\mathbf{X}^{(m)}_{t_0:t_L}\}_{m=1}^M,\ \epsilon,\ n_{\max}$}
    \State Construct the empirical quantities in the integrated least-squares formulation.
    \State Initialize $\widetilde A^{(0)}$ and $\theta^{(0)}$.
    \State Normalize the nonzero off-diagonal row blocks of $\widetilde A^{(0)}$ and normalize the diagonal vector $d(\widetilde A^{(0)})$ if nonzero.
    \For{$\tau=1,\dots,n_{\max}$}
        \State Update $\widetilde A^{(\tau)}$ rowwise by solving \eqref{eq:ials_A_update}, with $\theta=\theta^{(\tau-1)}$ fixed.
        \State Normalize the nonzero off-diagonal row blocks of $\widetilde A^{(\tau)}$ and normalize the diagonal vector $d(\widetilde A^{(\tau)})$ if nonzero.
        \State Update $\theta^{(\tau)}$ by solving \eqref{eq:ials_theta_update}, with $\widetilde A=\widetilde A^{(\tau)}$ fixed.
        \If{the relative changes of $\widetilde A$ and $\theta$ are both below $\epsilon$}
            \State \textbf{break}
        \EndIf
    \EndFor
    \State Extract $c^{(\tau)}$ and $\beta^{(\tau)}$ from $\theta^{(\tau)}$.
    \State \Return $\widetilde A^{(\tau)},\, c^{(\tau)},\, \beta^{(\tau)}$
\EndProcedure
\end{algorithmic}
\end{algorithm}

\subsection{Computational complexity}

We now compare the computational complexity of TALS and IALS at the level of one outer iteration. The purpose of this analysis is to clarify the algebraic cost of the subproblems induced by the two alternating schemes. Throughout, we focus on asymptotic complexity under explicit design-matrix assembly and standard dense least-squares solves.

Let $S:=ML$, $R:=dML=dS$, and $G:=dNML=NR$, where $N$ is the number of nodes, $M$ is the number of trajectories, $L$ is the number of observation intervals per trajectory, $d$ is the state dimension, $p$ is the number of interaction basis functions, and $q$ is the number of local basis functions. Here $R$ is the sample size of a single nodewise regression problem, while $G$ is the total number of scalar equations in the global least-squares subproblems.

To avoid committing to a specific NNLS solver, we denote by $\mathcal C_{\mathrm{NNLS}}(m,n)$ the cost of solving one nonnegative least-squares problem with $m$ rows and $n$ unknowns. The row-wise normalization steps in TALS and the split normalization in IALS require only vector rescalings and therefore contribute lower-order costs relative to the least-squares and NNLS solves.

\medskip
\noindent{\normalsize\bfseries{Per-iteration complexity of TALS.}}
In Algorithm~1, each outer iteration consists of three blocks.

First, the $A$-update is carried out row by row. For each $i\in[N]$, the subproblem \eqref{eq:A_nnls} is an NNLS problem with design matrix $Z_i(c)\in\mathbb R^{R\times (N-1)}$. Hence the total cost of the interaction-matrix update is $\mathcal O\!\big(N\,\mathcal C_{\mathrm{NNLS}}(R,N-1)\big)$.

Second, the interaction-coefficient update \eqref{eq:c_update} is a least-squares problem with design matrix $Z_c(A)\in\mathbb R^{G\times p}$. Under a dense direct solve, its cost is $\mathcal O(Gp^2+p^3)$.

Third, the local-coefficient update \eqref{eq:b_update} is a least-squares problem with design matrix $Z_b\in\mathbb R^{G\times q}$, whose cost is $\mathcal O(Gq^2+q^3)$.

Therefore, one outer iteration of TALS has complexity
\begin{equation}
\label{eq:complexity_tals}
\mathcal O\!\Big(
N\,\mathcal C_{\mathrm{NNLS}}(R,N-1)
+Gp^2+p^3
+Gq^2+q^3
\Big).
\end{equation}

\medskip
\noindent{\normalsize\bfseries{Per-iteration complexity of IALS.}}
In Algorithm~2, each outer iteration contains two blocks.

First, the diagonal-augmented matrix update \eqref{eq:ials_A_update} is again performed row by row. For each $i\in[N]$, the NNLS problem has design matrix $\widetilde Z_i(\theta)\in\mathbb R^{R\times N}$, so the total cost of this step is $\mathcal O\!\big(N\,\mathcal C_{\mathrm{NNLS}}(R,N)\big)$.

Second, the joint coefficient update \eqref{eq:ials_theta_update} is a least-squares problem with design matrix $\widetilde W(\widetilde A)\in\mathbb R^{G\times (p+q)}$. A dense direct solve therefore costs $\mathcal O\!\big(G(p+q)^2+(p+q)^3\big)$.

Hence one outer iteration of IALS has complexity
\begin{equation}
\label{eq:complexity_ials}
\mathcal O\!\Big(
N\,\mathcal C_{\mathrm{NNLS}}(R,N)
+G(p+q)^2+(p+q)^3
\Big).
\end{equation}

\medskip
\noindent{\normalsize\bfseries{Memory complexity.}}
If the global least-squares design matrices are assembled explicitly, then the storage required by the TALS coefficient steps is $\mathcal O(Gp+Gq)=\mathcal O(G(p+q))$, while the storage required by the IALS joint coefficient step is $\mathcal O(G(p+q))$.

In addition, if all interaction features are precomputed and stored explicitly for every sample, then the interaction tensor $B(X_{t_\ell}^{(m)})\in\mathbb R^{N\times N\times d\times p}$ leads to memory usage of order $\mathcal O(ML\,N^2dp)$. Similarly, explicit storage of all local feature tensors $H(X_{t_\ell}^{(m)})\in\mathbb R^{N\times d\times q}$ requires $\mathcal O(ML\,Ndq)$. For large-scale problems, this observation suggests that matrix-free or rowwise assembly may be preferable to full tensor storage.

\medskip
\noindent{\normalsize\bfseries{Complexity comparison and practical performance.}}
The above formulas yield a clear comparison between the two schemes in terms of per-iteration algebraic cost, while the numerical results also provide insight into their overall practical efficiency. For the coefficient updates, under dense direct least-squares solves, the leading algebraic cost of the joint update in IALS is strictly larger than the combined leading cost of the two separate coefficient updates in TALS whenever $p,q>0$. For the rowwise NNLS updates, if $\mathcal C_{\mathrm{NNLS}}(m,n)$ is
nondecreasing in the number of unknowns $n$, then the NNLS step in IALS is
also no cheaper than that in TALS.

Therefore, within the present algebraic complexity model, IALS does not enjoy a per-iteration cost advantage over TALS. Nevertheless, our numerical results indicate that IALS exhibits higher overall efficiency in most cases, in the sense that it typically converges in fewer outer iterations and often reaches the same, or even higher, estimation accuracy with less iteration effort; see Fig.~\ref{fig:figure1-2}. This suggests that the practical advantage of IALS should be understood as an iteration-efficiency gain rather than a reduction in the leading algebraic cost of a single iteration. A plausible explanation is that the integrated update in IALS couples the local and interaction coefficients more tightly, which tends to stabilize the alternating procedure and accelerate convergence toward an accurate solution.

\section{Identifiability}

In the diagonal-augmented formulation \eqref{eq:ials_model}, identifiability is understood under the split normalization convention used in Algorithm~\ref{alg:ials}: the off-diagonal interaction rows are normalized separately, while the diagonal entries are normalized collectively as a single local-scaling vector. Let $d(\widetilde A):=(\widetilde a_{11},\dots,\widetilde a_{NN})^\top$ and $\widetilde a_{i,-i}:=(\widetilde a_{ij})_{j\neq i}$. Define
\[
\widetilde{\mathcal M}
:=
\left\{
\widetilde A\in\mathbb R_+^{N\times N}:\,
\|d(\widetilde A)\|_2\in\{0,1\},\ 
\|\widetilde a_{i,-i}\|_2\in\{0,1\}\ \text{for all }i\in[N]
\right\}.
\]
For $(\widetilde A,H,\Phi)\in \widetilde{\mathcal M}\times \mathcal H_{\mathrm{loc}}\times \mathcal H_{\mathrm{int}}$, define the population loss
{\small
\begin{equation}
\label{eq:population_loss_ials}
\mathcal E_{L,\infty}(\widetilde A,H,\Phi)
:=
\frac{1}{L}\sum_{\ell=0}^{L-1}\sum_{i=1}^N
\mathbb E\!\left[
\left\|
\widetilde a_{ii} H(X_{t_\ell}^i)
+\sum_{j\neq i}\widetilde a_{ij}\Phi(X_{t_\ell}^j)
-\widetilde a_{ii}^* H_*(X_{t_\ell}^i)
-\sum_{j\neq i}\widetilde a_{ij}^* \Phi_*(X_{t_\ell}^j)
\right\|_2^2
\right].
\end{equation}
}
We say that $(\widetilde A_*,H_*,\Phi_*)$ is identifiable if it is the unique zero of \eqref{eq:population_loss_ials} over $\widetilde{\mathcal M}\times \mathcal H_{\mathrm{loc}}\times \mathcal H_{\mathrm{int}}$.

\begin{remark}
\label{rem:ials_normalization}
The split normalization in Algorithm~\ref{alg:ials} acts as a gauge-fixing convention. It removes the scaling freedom of the interaction block and the diagonal local-scaling vector separately: the off-diagonal rows are normalized independently, while the diagonal entries are normalized collectively. Consequently, the off-diagonal block can be interpreted as the interaction graph up to the imposed row normalization, whereas the diagonal block is interpreted only through the reconstructed local term.
\end{remark}

The above normalization fixes the trivial rescaling freedom between $\widetilde A$ and $(H,\Phi)$. We next introduce a coercivity condition ensuring that the population loss distinguishes different normalized parameter triples.

We say that \eqref{eq:ials_model} satisfies a rank-2 joint coercivity condition if there exists $c_{\mathcal H}>0$ such that, for every $i\in[N]$, every $H_1,H_2\in\mathcal H_{\mathrm{loc}}$ with $\langle H_1,H_2\rangle_{L^2(\rho_{\mathrm{loc}})}=0$, every $\Phi_1,\Phi_2\in\mathcal H_{\mathrm{int}}$ with $\langle \Phi_1,\Phi_2\rangle_{L^2(\rho_{\mathrm{int}})}=0$, and every $u^{(1)},u^{(2)}\in\mathbb R_+^N$,
\small{
\begin{equation}
\label{eq:rank2Condition}
\begin{aligned}
&\frac{1}{L}\sum_{\ell=0}^{L-1}
\mathbb E\!\left[
\left\|
u_i^{(1)}H_1(X_{t_\ell}^i)
+\sum_{j\neq i}u_j^{(1)}\Phi_1(X_{t_\ell}^j)
+u_i^{(2)}H_2(X_{t_\ell}^i)
+\sum_{j\neq i}u_j^{(2)}\Phi_2(X_{t_\ell}^j)
\right\|_2^2
\right]
\\
&\qquad\ge
c_{\mathcal H}\Big(
|u_i^{(1)}|^2\|H_1\|_{L^2(\rho_{\mathrm{loc}})}^2
+\|u_{-i}^{(1)}\|_2^2\|\Phi_1\|_{L^2(\rho_{\mathrm{int}})}^2
+|u_i^{(2)}|^2\|H_2\|_{L^2(\rho_{\mathrm{loc}})}^2
+\|u_{-i}^{(2)}\|_2^2\|\Phi_2\|_{L^2(\rho_{\mathrm{int}})}^2
\Big),
\end{aligned}
\end{equation}
}
where $u_{-i}:=(u_j)_{j\neq i}$.

\begin{theorem}[Identifiability]
\label{thm:ials_identifiability}
Let $(\widetilde A_*,H_*,\Phi_*)\in \widetilde{\mathcal M}\times \mathcal H_{\mathrm{loc}}\times \mathcal H_{\mathrm{int}}$ with $H_*\neq0$ and $\Phi_*\neq0$. Assume that \eqref{eq:rank2Condition} holds, and that there exist indices $i_{\mathrm{loc}},i_{\mathrm{int}}\in[N]$ for which $\widetilde a_{i_{\mathrm{loc}}i_{\mathrm{loc}}}^*>0$ and $\|\widetilde a_{i_{\mathrm{int}},-i_{\mathrm{int}}}^*\|_2>0$. Then $(\widetilde A_*,H_*,\Phi_*)$ is identifiable.
\end{theorem}

A proof is given in Appendix~\ref{app:proof_ials_identifiability}.

\section{Numerical experiments}

\subsection{Synthetic experiments}
\subsubsection{Lennard--Jones-type interaction example}

\medskip
\noindent{\normalsize\bfseries{Experiment setup.}}
We first describe the common setup used throughout the simulated-data experiments. There are $N=10$ agents in a fully interacting network. The ground-truth interaction matrix is denoted by $A_*=(a_{ij}^*)_{1\le i,j\le N}$, with zero diagonal. Its off-diagonal entries are sampled independently from the uniform distribution on $[0,1]$ and then row-normalized, i.e., $a_{ii}^*=0$, $a_{ij}^*\in[0,1]$ for $j\neq i$, and $\sum_{j\neq i}(a_{ij}^*)^2=1$ for each $i\in[N]$. The state variable is one-dimensional, namely $X_t^i\in\mathbb R$.

The true interaction kernel and local drift are chosen as
\[
\Phi_*(r)=
\begin{cases}
-\dfrac{1}{3}r^{-9}+\dfrac{4}{3}r^{-3}, & r\ge 0.5,\\[0.3em]
-160, & 0\le r<0.5,
\end{cases}
\]
where $r=\|X_t^j-X_t^i\|_2$ denotes the pairwise distance between two agents. The true local drift is given by
\[
L_*(x)=-100x+625x^3.
\]

We approximate the interaction kernel by
\[
\Phi(r)=\sum_{k=1}^{p}c_k\psi_k(r)
\]
using the enlarged piecewise dictionary
\[
\Psi=
\left\{
\begin{aligned}
&\psi_{1+k}(r)=r^{-9}\mathbf{1}_{[0.25k+0.5,\infty)}(r), && k=0,1,2,\\
&\psi_{4+k}(r)=r^{-3}\mathbf{1}_{[0.25k+0.5,\infty)}(r), && k=0,1,2,\\
&\psi_{7+k}(r)=\mathbf{1}_{[0,\,0.25k+0.5]}(r), && k=0,1,2,3.
\end{aligned}
\right.
\]
This dictionary is larger than necessary and is therefore treated as potentially misspecified. In this representation, the true coefficient vector $c^*$ is zero except for $(c_1^*,c_4^*,c_7^*)=\left(-\tfrac{1}{3},\tfrac{4}{3},-160\right)$.

We approximate the local drift by
\[
L(x)=\sum_{k=1}^{q}b_k\varphi_k(x),
\qquad
\varphi_k(x)=x^k,\quad k=1,\dots,5,
\]
so that $q=5$. This polynomial dictionary is again larger than necessary and may therefore be viewed as potentially misspecified. In this basis, the true coefficient vector $b^*$ is zero except for $(b_1^*,b_3^*)=(-100,625)$. No sparsity is imposed in the estimation procedure.

The multi-trajectory synthetic data are generated by the Euler--Maruyama scheme with time step $\Delta t=10^{-4}$. The initial condition $\mathbf X_{t_0}=(X_{t_0}^i,\ i=1,\dots,N)$ is sampled component-wise from an initial distribution $\mu_0$. At the observation times, we record
\[
\mathbf Y_{t_\ell}^{(m)}=\mathbf X_{t_\ell}^{(m)}+\boldsymbol{\varepsilon}_{t_\ell}^{(m)},
\qquad
\boldsymbol{\varepsilon}_{t_\ell}^{(m)}\sim \mathcal N(0,\tau_{\mathrm{obs}}^2 I),
\]
where the observation noises are taken to be independent across trajectories and time indices. When $\tau_{\mathrm{obs}}=0$, the observations are noiseless. Other parameters will be specified in each individual experiment.

\begin{figure}[htbp]
    \centering
    \includegraphics[width=0.9\linewidth]{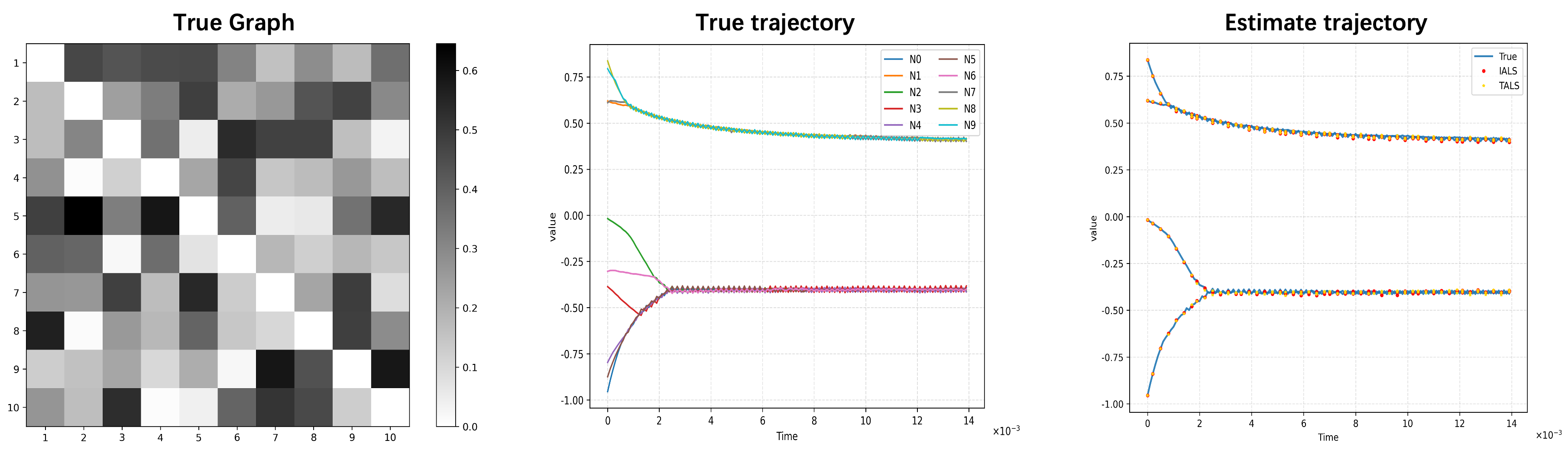}
    \caption{\textbf{System and trajectory reconstruction example.} Left: ground-truth interaction matrix $A_*$. Middle: true trajectories. Right: trajectory reconstructions produced by IALS and TALS.}
    \label{fig:figure1-1}
\end{figure}

\medskip
\noindent{\normalsize\bfseries{ Reconstruction example.}}

We begin with a representative example illustrating the reconstruction quality and convergence behavior of the two estimators. In this test, the initial distribution is $\mu_0=\mathrm{Unif}([-1,1])$, the training set consists of $M=500$ trajectories, the stochastic forcing intensity is $\sigma=10^{-3}$, and the observation-noise standard deviation is $\tau_{\mathrm{obs}}=10^{-3}$. The data are sampled with time step $\Delta t=10^{-4}$ over $L=300$ observation times, corresponding to the horizon $T=L\Delta t=0.03$.

Figure~\ref{fig:figure1-1} shows the ground-truth interaction matrix together with the true trajectories and the corresponding trajectory reconstructions produced by IALS and TALS. Both methods reproduce the trajectory evolution accurately in this representative run.

Figure~\ref{fig:figure1-2} further reports the recovered interaction kernel and local drift, the entrywise error heat maps of the interaction matrices estimated by IALS and TALS, and the convergence histories of the two methods. Both estimators recover the interaction structure and the dynamical components accurately. Computationally, however, IALS exhibits markedly faster convergence: within six iterations, it achieves accuracy comparable or superior to that of TALS after far more iterations. Since the two methods exhibit similar reconstruction quality here while IALS is markedly more computationally efficient, the subsequent sensitivity and sample-size studies focus primarily on IALS.

\begin{figure}[htbp]
    \centering
    \includegraphics[width=0.9\linewidth]{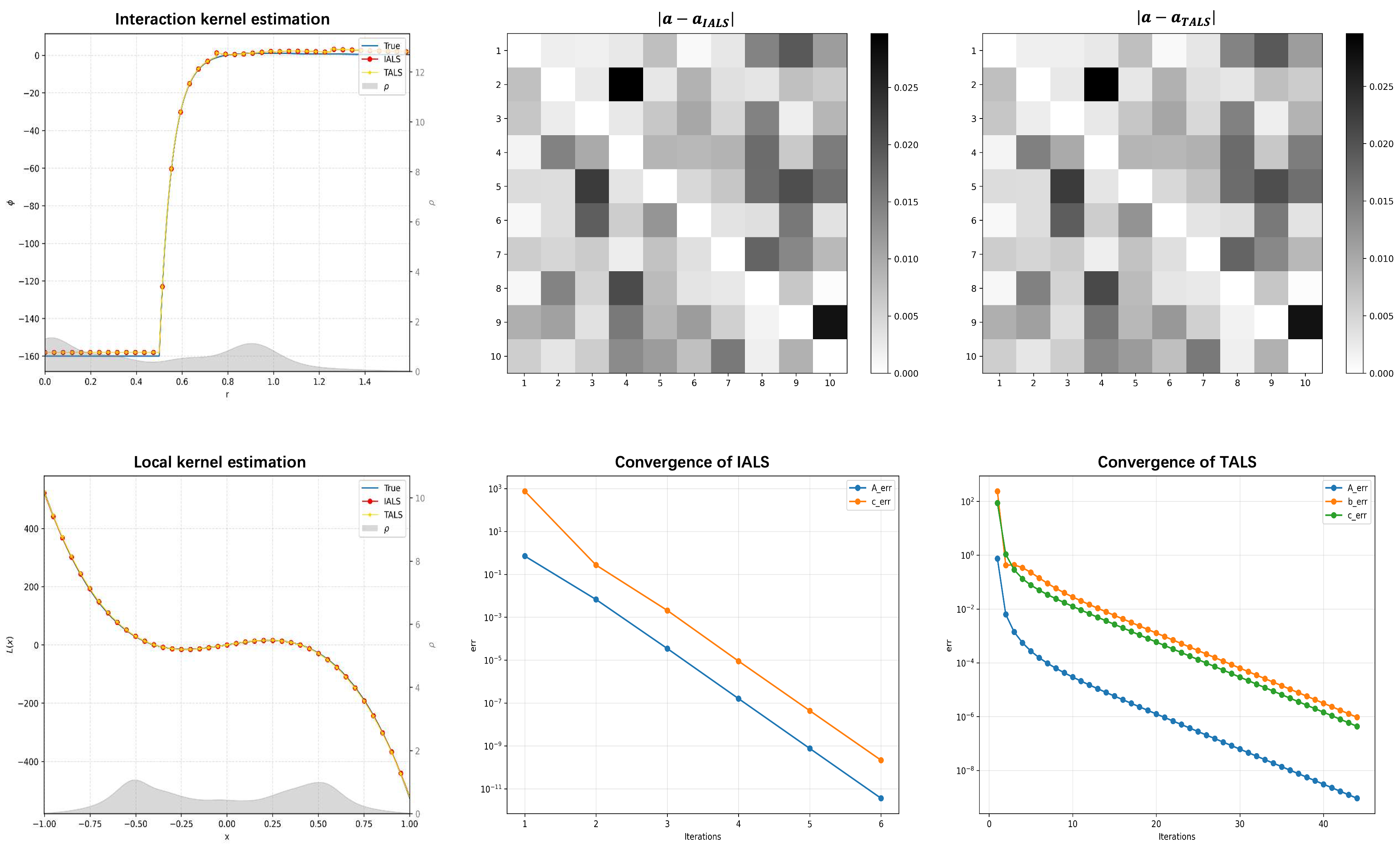}
    \caption{\textbf{Recovery of dynamical components, graph-estimation errors, and convergence histories.} Top left: recovery of the interaction kernel $\Phi$ by IALS and TALS. Top middle and top right: entrywise error heat maps of the interaction matrices estimated by IALS and TALS, respectively. Bottom left: recovery of the local drift $L$ by IALS and TALS. Bottom middle and bottom right: convergence histories of IALS and TALS, respectively.}
    \label{fig:figure1-2}
\end{figure}

\medskip
\noindent{\normalsize\bfseries{Convergence with respect to sample size.}}

We next investigate how the performance of IALS depends on the number of training trajectories $M$, and how this behavior is affected by observational noise and basis misspecification. To this end, we consider two settings.

\begin{enumerate}
    \item \textbf{Noiseless data with a well-specified basis.}
    We use noiseless observations together with the exact dictionaries $\{\psi_1,\psi_4,\psi_7\}$ and $\{\varphi_1,\varphi_3\}$. This setting is intended to examine whether the estimator exhibits the expected parametric convergence behavior as $M$ increases.

    \item \textbf{Noisy data with enlarged dictionaries.}
    We set $\tau_{\mathrm{obs}}=10^{-2}$ and use the enlarged dictionaries $\{\psi_k\}_{k=1}^{7}$ and $\{\varphi_k\}_{k=1}^{5}$. This setting is used to assess robustness with respect to noise and potential basis misspecification.
\end{enumerate}

\begin{figure}[!htbp]
    \centering
    \includegraphics[width=0.9\linewidth]{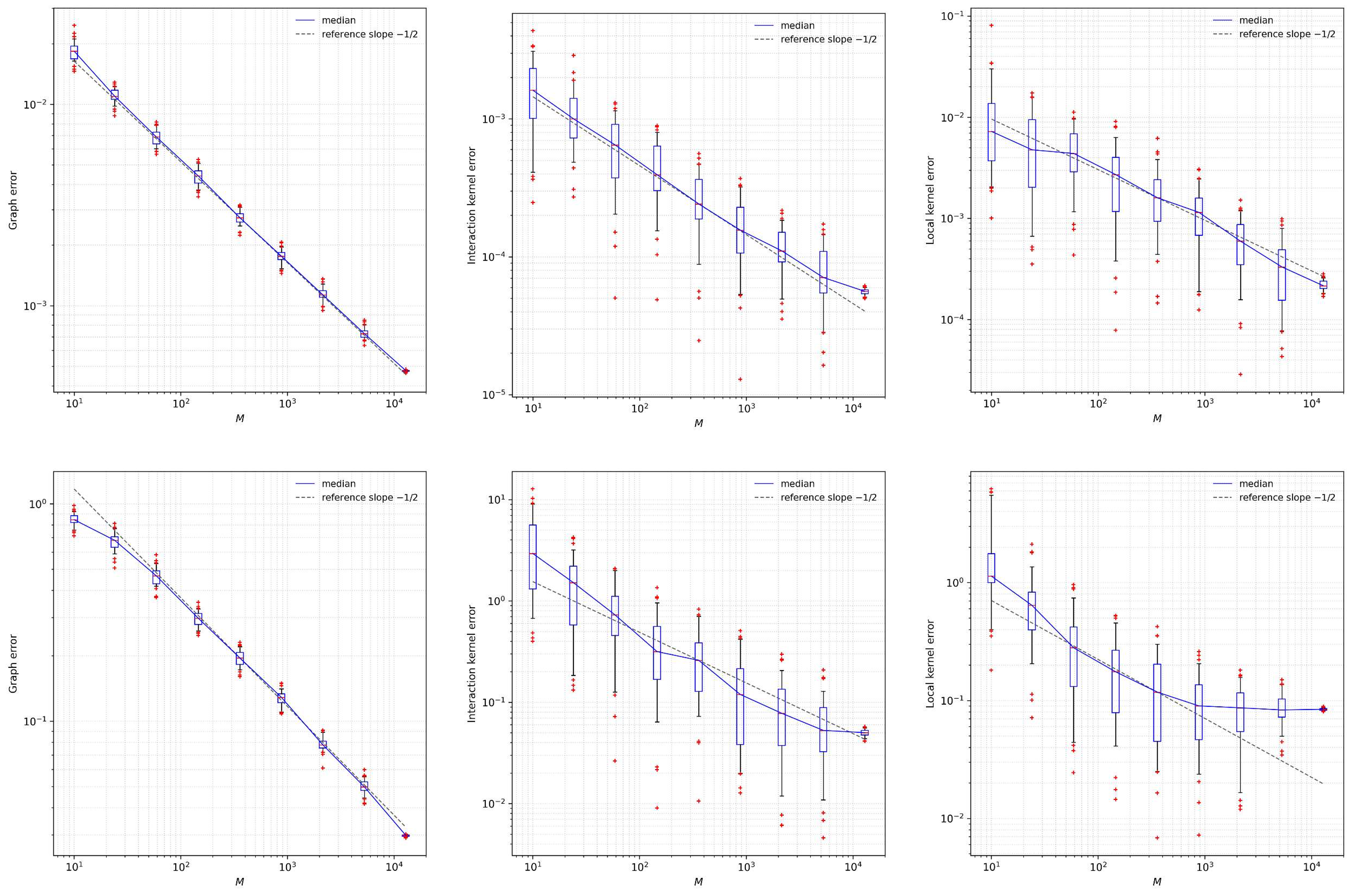}
    \caption{Convergence of IALS as the sample size $M$ increases over 50 independent runs. \textbf{Top row:} noiseless data with a well-specified basis; all three errors decay approximately at the parametric rate of $-1/2$ \textbf{Bottom row:} noisy data with enlarged dictionaries; the graph error continues to decay with $M$, whereas the interaction- and local-drift errors decrease more slowly and eventually plateau.}
    \label{fig:figure2}
\end{figure}

Figure~\ref{fig:figure2} summarizes the results over 50 independent runs. In the noiseless and well-specified setting (top row), all three errors decrease steadily as $M$ increases. On the log--log scale, the decay is close to linear, with slope near $-1/2$, which is consistent with the expected parametric rate in a correctly specified finite-dimensional setting.

In the noisy and enlarged-dictionary setting (bottom row), the graph error still decreases markedly with $M$, and its decay rate on the log--log scale remains approximately $-1/2$, consistent with the parametric rate observed in the noiseless setting. This indicates that the interaction structure remains identifiable from additional trajectories. By contrast, both kernel errors decrease more slowly and eventually level off, with visibly larger dispersion across repetitions. This suggests that, once the sample size is sufficiently large, kernel recovery becomes limited by a noise- and misspecification-dominated error floor, so that further increases in $M$ yield only marginal improvement. Overall, the results indicate that IALS retains good graph-recovery accuracy under moderate noise and basis enlargement, while the kernel errors are more sensitive to these perturbations.

\medskip
\noindent{\normalsize\bfseries{Robustness to stochastic forcing and observation noise.}}

We next examine the robustness of IALS with respect to stochastic forcing and observation noise by varying the corresponding noise levels over several orders of magnitude. Throughout this study, all results are summarized over 50 independent Monte Carlo runs.

We consider two complementary settings. In the first one, we vary the stochastic forcing intensity $\sigma$ while fixing the observation-noise standard deviation at $\tau_{\mathrm{obs}}=10^{-7}$, the sample size at $M=1000$, and the trajectory length at $L=100$. This setting is intended to isolate the effect of stochastic forcing on the estimator. In the second one, we set $\sigma=0$ and vary the observation-noise standard deviation $\tau_{\mathrm{obs}}$, again with $M=1000$, in order to assess the sensitivity of the estimator to measurement noise alone.

Figure~\ref{fig:robustness} reports the resulting graph, interaction-kernel, and local-drift errors. In the top row, all three errors decrease as the stochastic forcing intensity is reduced. For moderate to small values of $\sigma$, the decay is close to linear on the log--log scale, indicating an approximately first-order dependence on the forcing level over the resolved range. The graph error exhibits the clearest trend, while the two kernel errors show slightly larger variability, especially at very small forcing levels.

The bottom row shows a similar trend with respect to the observation-noise standard deviation. As $\tau_{\mathrm{obs}}$ decreases, all three errors decrease as well, with the graph error again displaying the most regular scaling. The interaction- and local-drift errors remain more dispersed across Monte Carlo runs, reflecting their higher sensitivity to finite-sample variability and noise. Overall, these experiments indicate that IALS is robust to both stochastic forcing and observation noise: graph recovery remains accurate across a wide range of noise levels, while kernel recovery degrades more gracefully and improves steadily as the noise is reduced.

\begin{figure}[!htbp]
    \centering
    \includegraphics[width=0.9\linewidth]{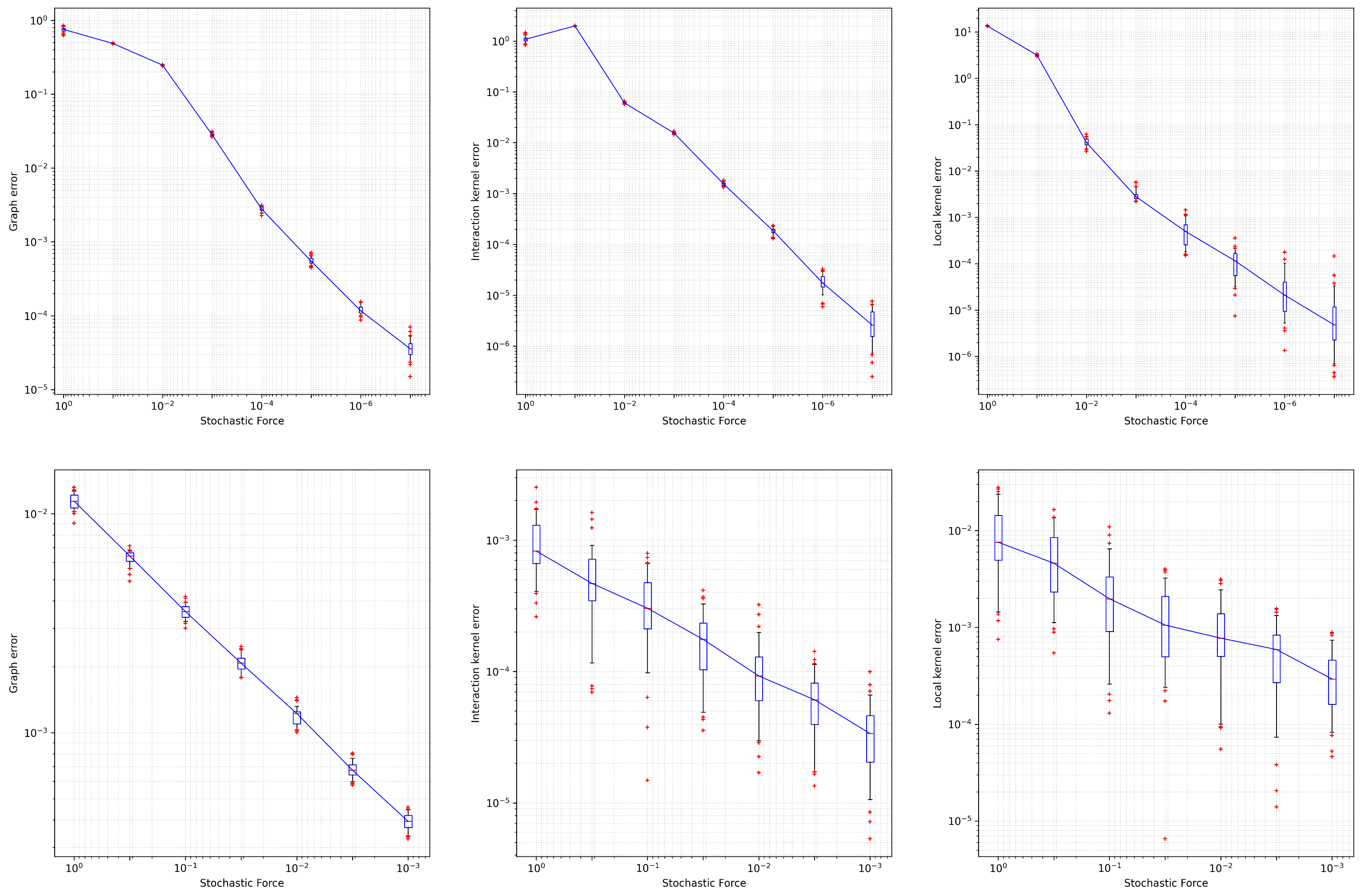}
    \caption{Robustness of IALS to stochastic forcing and observation noise over 50 independent runs. \textbf{Top row:} graph, interaction-kernel, and local-drift errors versus the stochastic forcing intensity $\sigma$, with fixed observation-noise standard deviation $\tau_{\mathrm{obs}}=10^{-7}$, sample size $M=1000$, and trajectory length $L=100$. \textbf{Bottom row:} the same three errors versus the observation-noise standard deviation $\tau_{\mathrm{obs}}$, with $\sigma=0$ and $M=1000$. In each panel, red markers denote individual trials and the blue curve connects the medians across noise levels.}
    \label{fig:robustness}
\end{figure}

\subsubsection{Comparison under model mismatch}

We next examine the role of the local term by comparing IALS with the standard ALS baseline in a controlled one-dimensional example. The purpose of this experiment is twofold: first, to assess whether explicitly modeling the local drift improves structural recovery when the true dynamics contain a local component; second, to verify that the additional flexibility of IALS does not introduce spurious local effects when the true dynamics are interaction-only.

The ground-truth model is specified by
\[
L_*(x)=-0.5x+0.5x^3,
\qquad
\Phi_*(r)=3\sin(r)-1.5\sin(3r).
\]
The initial states are sampled from $(-\pi,\pi)$ using a rejection scheme that favors points with larger drift magnitude, so as to avoid weakly informative initializations near regions where the drift is close to zero. Throughout this experiment, we set $\tau_{\mathrm{obs}}=0$ and fix the stochastic forcing intensity to $\sigma=10^{-5}$.

We consider two data-generation regimes. In the \emph{with-local} regime, the training data are generated from the full model containing both $L_*$ and $\Phi_*$. In the \emph{without-local} regime, the local term is removed and the data are generated from the interaction-only model. In each regime, we train both ALS and IALS on the corresponding data and compare the learned graph structure and the recovered dynamics.

Figure~\ref{fig:figure4} summarizes the comparison. In the \emph{with-local} regime, IALS recovers a graph that is consistent with the ground truth. By contrast, the learned graph from ALS is noticeably distorted, indicating that the local effect is absorbed into an incorrect interaction structure. In the \emph{without-local} regime, IALS continues to recover the correct dynamics and, importantly, drives the estimated local term to a quantity that is numerically close to zero. This shows that introducing the local component in the model does not compromise performance when the true system does not contain such a term.

Taken together, these results indicate that IALS remains stable across both regimes: it improves structural interpretability when a local term is present, while avoiding spurious local reconstruction when the true dynamics are interaction-only. In contrast, ALS is reliable only when the local term is absent; otherwise, it may recover an incorrect dynamical mechanism.

\begin{figure}[!htbp]
    \centering
    \includegraphics[width=0.9\linewidth]{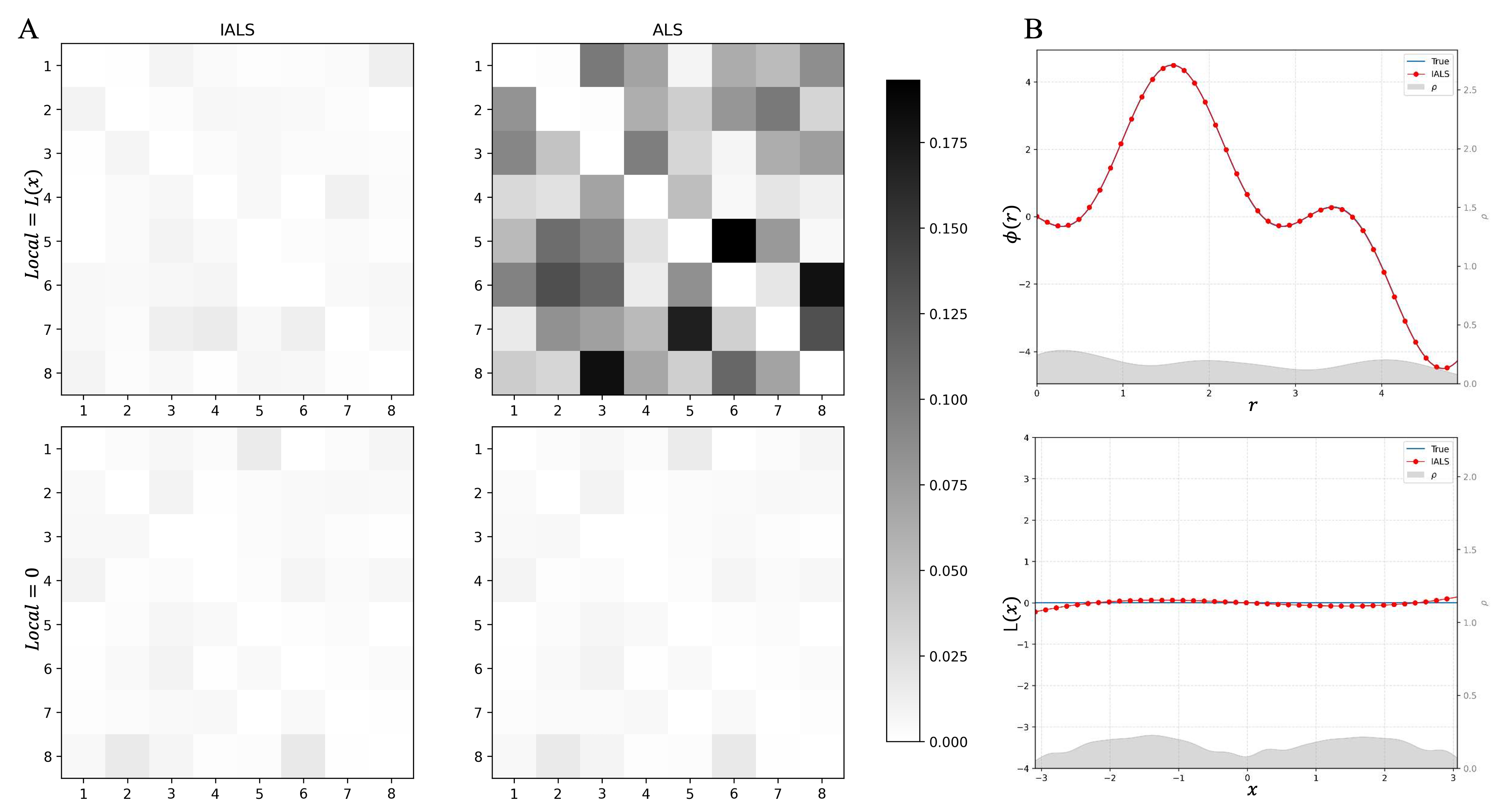}
    \caption{Comparison of graph and dynamics recovery between ALS and IALS with ($\text{Local=L(x)}$) and without ($\text{Local=0}$) a local term. \textbf{(A)} Entrywise error heatmaps of the recovered graph matrices for ALS and IALS under both regimes; ALS exhibits significantly larger graph recovery errors in the with-local regime. \textbf{(B)} Dynamics recovery by IALS in the without-local regime: (top) interaction kernel recovery; (bottom) local term recovery, which is driven numerically close to zero as expected. }
    \label{fig:figure4}
\end{figure}

\subsection{Application to ictal SEEG recordings}

We finally apply the proposed IALS framework to real ictal SEEG recordings in order to assess its behavior beyond controlled synthetic settings. Unlike the simulated experiments above, this application does not provide a ground-truth interaction graph or ground-truth kernels. Accordingly, the goal here is not exact recovery, but rather to examine whether the learned dynamics are internally consistent with the observed seizure trajectories, whether the resulting graph-derived functional summaries are clinically interpretable, and whether these summaries remain stable across different dictionary choices.

\subsubsection{Experimental setup}
We use ictal and interictal SEEG recordings from the ds4100 dataset organized in BIDS format. Channels are restricted to those labeled as \texttt{good}, and only channels shared by the interictal and ictal recordings are retained to ensure a consistent baseline normalization. Seizure onset and offset are identified from the event annotations, and for each seizure we analyze a padded window around the ictal episode. From the multichannel voltage traces, we extract high-frequency amplitude features by band-pass filtering in 80--150\,Hz, computing the Hilbert envelope, and forming a log-power feature, followed by temporal smoothing. To improve comparability across subjects, all feature trajectories are represented on a common sampling grid, and channel-wise z-score normalization is performed using interictal baseline statistics. The resulting multivariate feature trajectories are used as the state observations for dynamical identification. 

To assess robustness with respect to the functional parameterization, we consider four paired configurations for the local and interaction terms,
\[
(L,\Phi)\in\{(L_1,\Phi_1),(L_2,\Phi_2),(L_3,\Phi_3),(L_4,\Phi_4)\},
\]
with
\begin{itemize}
  \item $(L_1,\Phi_1)=(\texttt{tanh\_zscore\_5x2},\,\texttt{Polynomial+activation})$,
  \item $(L_2,\Phi_2)=(\texttt{5-Fourier},\,\texttt{5-Fourier})$,
  \item $(L_3,\Phi_3)=(\texttt{5-Polynomial},\,\texttt{5-Polynomial})$,
  \item $(L_4,\Phi_4)=(\texttt{LocalSmallStrong},\,\texttt{PhiPhysio})$.
\end{itemize}
For brevity, we refer to these four paired dictionary configurations simply as \texttt{Poly+Act}, \texttt{5-Fourier}, \texttt{5-Polynomial}, and \texttt{PhiPhysio}, respectively, using the label of the interaction dictionary as a shorthand for the full pair. For each patient and each dictionary pair, we fit a patient-specific dynamical model and obtain a learned interaction matrix together with the associated basis coefficients. The explicit functional forms and hyper-parameter settings of the dictionaries are given in ~\ref{app:dictionary_specifications}.

The application results are presented at several complementary levels. We first provide an illustrative single-subject visualization, including the learned interaction matrix together with the fitted local and interaction functions. We then examine predictive consistency both at the subject level, through node-wise prediction results, and at the cohort level, through cross-dictionary comparisons of predictive performance. Finally, we analyze SOZ-related directional-flow quantities and assess their robustness.

\subsubsection{Representative single-subject visualization}

We first present the fitted outputs for an illustrative patient, HUP188, to demonstrate the behavior of the proposed framework on real SEEG data; see Fig.~\ref{fig:figure5}. The model yields a learned channel-level interaction matrix together with the fitted interaction function and example node-specific local terms. In the IALS formulation, the diagonal entries of the learned matrix are interpreted as node-dependent scaling factors for the shared local term rather than as self-loops; their visible heterogeneity suggests that the local component is used in a nontrivial channel-dependent manner.

The learned interaction matrix is clearly nonuniform and approximately sparse, despite the absence of explicit sparsity regularization. In the context of epileptic network propagation, such a nonuniform and relatively selective pattern appears plausible, as focal seizures are increasingly understood as emerging from and spreading through epileptogenic networks rather than uniformly across all recorded regions\cite{bartolomei2017defining,Panzica2013IdentificationOT}.

The fitted local and interaction functions also remain nontrivial, indicating that the model does not collapse to a purely graph-based representation. Although no ground truth is available in this real-data setting, the learned functions appear qualitatively reasonable at a coarse-grained level. In particular, the interaction function suggests a state-difference-dependent nonlinear coupling pattern, which may be viewed as reflecting that the effective influence between channels varies with their activity mismatch rather than remaining constant. The local function, on the other hand, appears to provide a nonlinear regulating or restoring effect on channel activity, limiting excessively large deviations while allowing nontrivial local dynamics. Taken together, these fitted components are broadly compatible with the view that seizure evolution is shaped jointly by network-mediated recruitment and local nonlinear regulation. 
\begin{figure}[!htbp]
    \centering
    \includegraphics[width=0.9\linewidth]{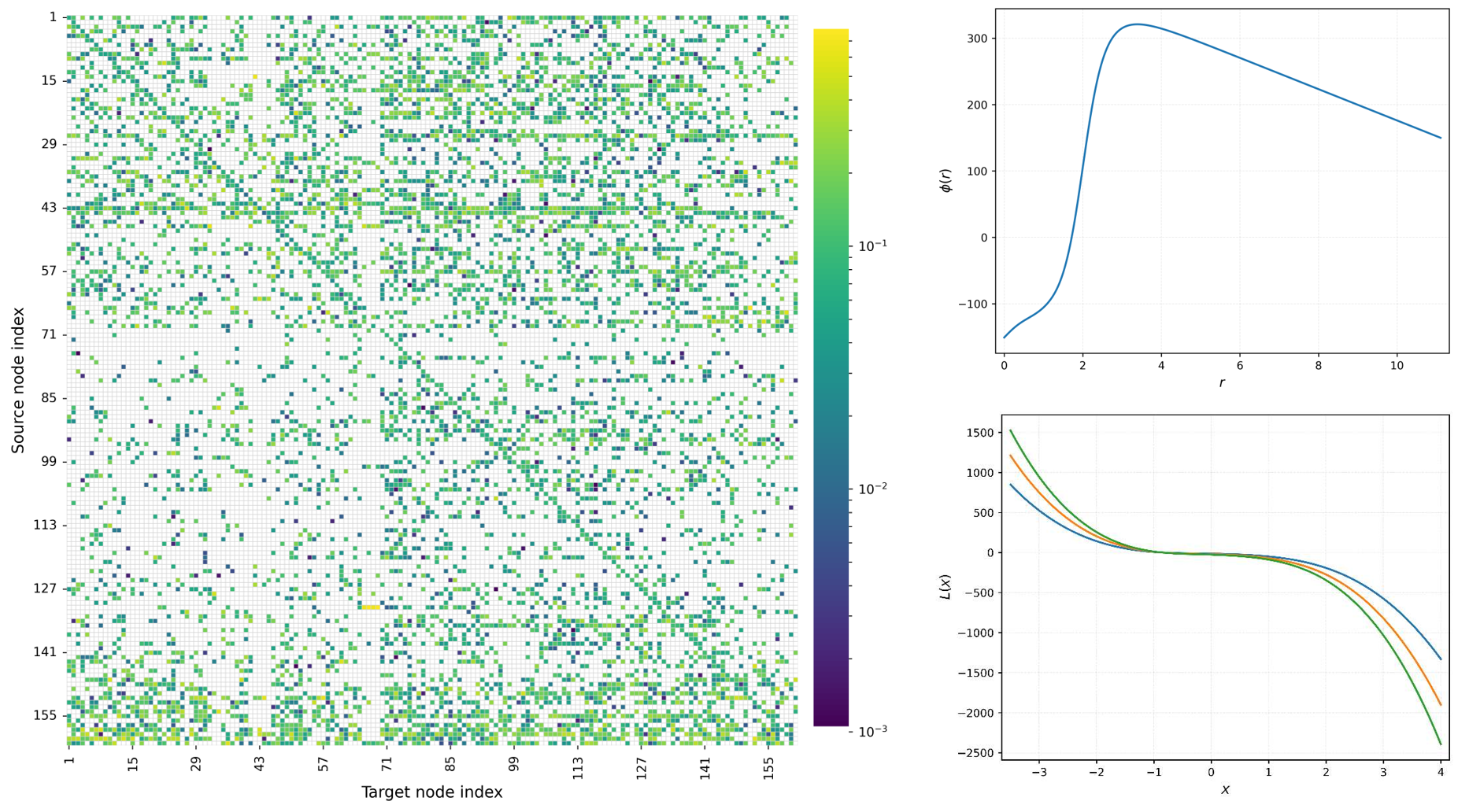}
    \caption{\textbf{Fitted outputs on real SEEG data for patient HUP188.}
Left: learned channel-level interaction matrix. Right: fitted interaction function (upper panel) and three example node-specific local terms (lower panel). In the IALS formulation, the diagonal entries of the learned matrix act as node-dependent scaling factors for the shared local term. Accordingly, the lower-right panel displays three example local terms induced by different diagonal scalings.}
    \label{fig:figure5}
\end{figure}

\subsubsection{Prediction diagnostics and cross-dictionary consistency}

To evaluate how well the learned dynamics remain locally consistent with the observed ictal trajectories, we examine prediction errors from three complementary perspectives: node-wise diagnostics for an illustrative patient, patient-level cross-dictionary comparisons based on mean node-wise errors, and variability across patients. At the cohort level, we summarize predictive performance using NRMSE$_{\mathrm{std}}$, NMAE$_{\mathrm{std}}$, and sMAPE (Table~\ref{tab:kernel_summary_compact}). Here NRMSE$_{\mathrm{std}}$ and NMAE$_{\mathrm{std}}$ denote the RMSE and MAE normalized by the empirical standard deviation of the pooled prediction targets, respectively, while sMAPE denotes the symmetric mean absolute percentage error. For the node-wise analysis, we additionally report a raw relative MAE and a $\tau$-stabilized relative MAE, denoted by RMAE$_{\mathrm{raw}}$ and RMAE$_\tau$, together with the node-wise MAE and prediction bias. The precise definitions of all metrics, including the choice of the stabilization threshold $\tau$, are provided in Appendix~\ref{app:metric_definitions}.

\begin{figure}[!htbp]
    \centering
    \includegraphics[height=0.8\textheight, keepaspectratio]{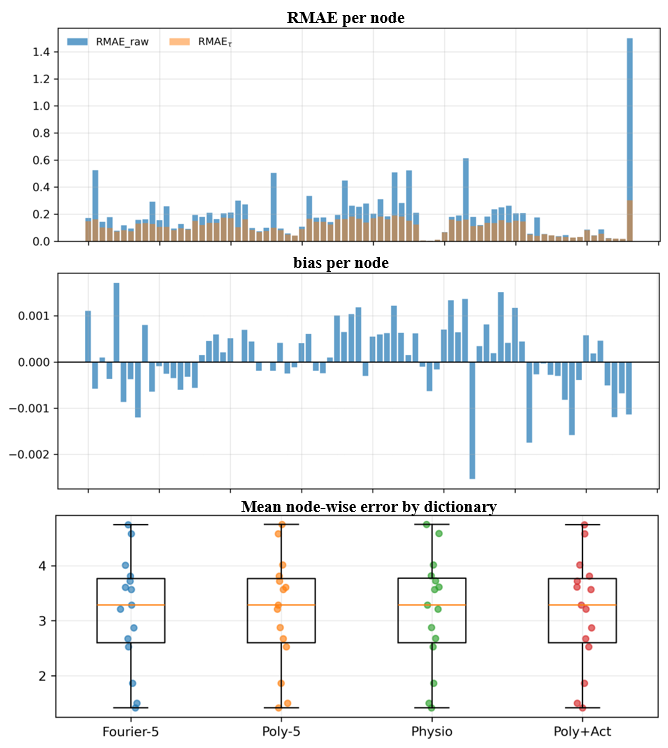}
    \caption{\textbf{Prediction diagnostics and cross-dictionary stability.}
    \textbf{Top:} node-wise relative prediction errors for an illustrative patient, including both RMAE$_{\mathrm{raw}}$ and the $\tau$-stabilized RMAE$_\tau$. \textbf{Middle:} node-wise prediction bias for the same patient. \textbf{Bottom:} patient-level boxplots of the average node-wise errors, where each point corresponds to the mean error over all nodes for one patient under a given dictionary configuration. The broadly similar distributions across dictionaries suggest that predictive performance remains largely stable with respect to dictionary choice.}
    \label{fig:figure6}
\end{figure}

Figure~\ref{fig:figure6} first reports the node-wise prediction errors and biases for an illustrative patient, highlighting that predictive performance varies across channels. It then summarizes, through boxplots, the distributions of patient-level mean node-wise errors under the four dictionary configurations. The resulting distributions are broadly similar across dictionaries, suggesting that predictive performance remains largely stable with respect to dictionary choice, whereas variability across patients is comparatively more pronounced.

The results in Table~\ref{tab:kernel_summary_compact} further support this observation. No single dictionary is uniformly preferred across patients, while the reported normalized errors remain at the percent level. Given the substantial artifact contamination, heterogeneity, and temporal variability of ictal SEEG recordings \cite{kharbouch2011algorithm,bush2022differentiation}, such error levels are still encouraging and indicate reasonably good local predictive consistency between the learned dynamics and the observed trajectories.

\begin{table}[!htbp]
\centering
\caption{\textbf{Cohort-level prediction errors across dictionary configurations.}
Values are reported in percentage form as median [25th, 75th percentile] across patients ($N=22$).}
\label{tab:kernel_summary_compact}
\small
\setlength{\tabcolsep}{5pt}
\begin{tabular}{lccc}
\toprule
Dictionary & NRMSE$_{\mathrm{std}}$ (\%) & NMAE$_{\mathrm{std}}$ (\%) & sMAPE (\%) \\
\midrule
5-Fourier     & 3.043 [1.734, 3.693] & 2.293 [1.319, 2.775] & 1.857 [0.952, 2.819] \\
5-Polynomial  & 3.045 [1.733, 3.694] & 2.295 [1.318, 2.776] & 1.856 [0.952, 2.821] \\
PhiPhysio     & 3.045 [1.732, 3.696] & 2.296 [1.318, 2.779] & 1.856 [0.952, 2.822] \\
Poly+Act      & 3.044 [1.732, 3.696] & 2.294 [1.318, 2.778] & 1.857 [0.952, 2.819] \\
\bottomrule
\end{tabular}
\end{table}

\subsubsection{SOZ-related directional flow and cross-dictionary robustness}

After establishing adequate local predictive consistency, we analyze SOZ--NA directional-flow summaries induced by the learned off-diagonal interaction matrix. For each node $v$, we define the net flow by
\[
\mathrm{net}(v)=\mathrm{in}(v)-\mathrm{out}(v),
\]
where, for an SOZ node, $\mathrm{in}(v)$ denotes the total incoming influence from NA nodes and $\mathrm{out}(v)$ denotes the total outgoing influence toward NA nodes; the definition is symmetric for NA nodes. Under this sign convention, $\mathrm{net}(v)<0$ indicates that the node exerts stronger outward influence across the SOZ--NA partition than it receives, and thus behaves as a more source-like node with respect to the opposite partition.

This quantity is of direct biological interest in the ictal setting. Under the sign convention adopted here, a more negative net flow indicates stronger outward influence across the SOZ--NA partition. Since seizure-onset regions are expected to drive activity toward non-SOZ tissue during ictal propagation, more negative net flow in SOZ is consistent with this biological expectation \cite{bartolomei2017defining,korzeniewska2014ictal,matarrese2023spike}. Therefore, SOZ--NA net flow provides a physiologically interpretable summary of the learned dynamics.

We first examine the sign and magnitude of these SOZ-related net-flow quantities at the cohort level. As summarized in Table~\ref{tab:soz_net_sign_summary} and visualized in Fig.~\ref{fig:figure7}, the SOZ net-flow quantities tend to be negative in a substantial proportion of subjects, indicating a systematic tendency toward stronger outward influence across the SOZ--NA cut. Although no ground-truth interaction graph is available in the real SEEG setting, this finding is consistent with the expected source-like role of SOZ during ictal propagation.

We then examine whether this biologically meaningful quantity is stable with respect to the choice of basis functions. Although the four dictionary configurations can produce visibly different interaction matrices, the induced per-node NA net-flow summaries remain considerably more consistent within each subject. Figure~\ref{fig:figure8} illustrates this phenomenon for an illustrative subject: the learned matrices differ appreciably at the edge-weight level, yet the per-node NA net-flow curves across dictionaries strongly overlap, and the pairwise scatter plots concentrate near the diagonal $y=x$. This indicates that the biological interpretation carried by the net-flow summary is not an artifact of a particular dictionary configuration.

Taken together, these results suggest that SOZ-related directional flow is both biologically meaningful and empirically robust. In particular, the tendency of SOZ nodes to appear more source-like is preserved across different basis choices, which strengthens the credibility of this interpretation and supports the use of net flow as a clinically relevant summary of the learned seizure dynamics.

\begin{table}[t]
\centering
\caption{Cohort-level summary of SOZ net-flow sign across dictionaries. Negative $\mathrm{net}$ indicates dominant outward influence across the SOZ--NA partition.}
\label{tab:soz_net_sign_summary}
\begin{tabular}{lcc}
\toprule
Metric & Count & Percentage \\
\midrule
Subjects & 22& -- \\
SOZ net $<0$ in $\ge 1$ dictionary & 16/22& 73\%\\
SOZ net $<0$ in all dictionaries & 15/22& 68\%\\
Subject-averaged SOZ net $<0$ (mean over dictionaries) & 16/22& 73\%\\
\bottomrule
\end{tabular}
\end{table}

\begin{figure}[!htbp]
    \centering   \includegraphics[width=0.9\linewidth,height=0.3\textheight]{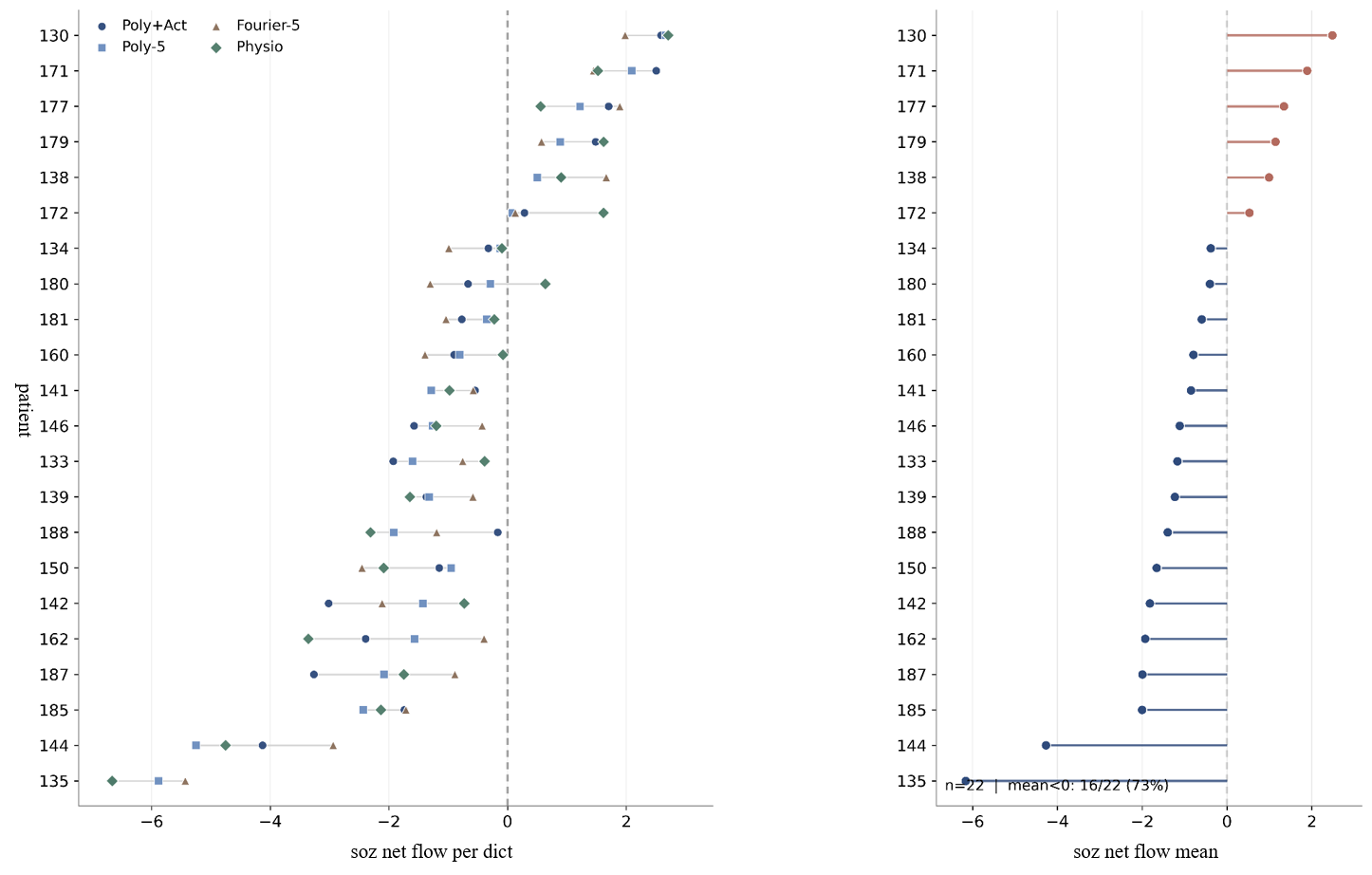}
   \caption{\textbf{SOZ-related net-flow summaries across subjects and dictionary configurations.}
For each subject, the SOZ net-flow quantity is computed from the learned interaction matrix as $\mathrm{net}(v)=\mathrm{in}(v)-\mathrm{out}(v)$, where negative values indicate a more source-like role across the SOZ--NA partition. \textbf{Left:} subject-level SOZ net flow under the four dictionary configurations; each row corresponds to one subject, and the four markers are connected to show the within-subject variation across dictionaries. \textbf{Right:} subject-level mean SOZ net flow obtained by averaging over the four dictionary configurations. Subjects are ordered by their mean SOZ net flow, from more negative to more positive values.}
    \label{fig:figure7}
\end{figure}

\begin{figure}[!htbp]
    \centering
    \includegraphics[width=0.9\linewidth]{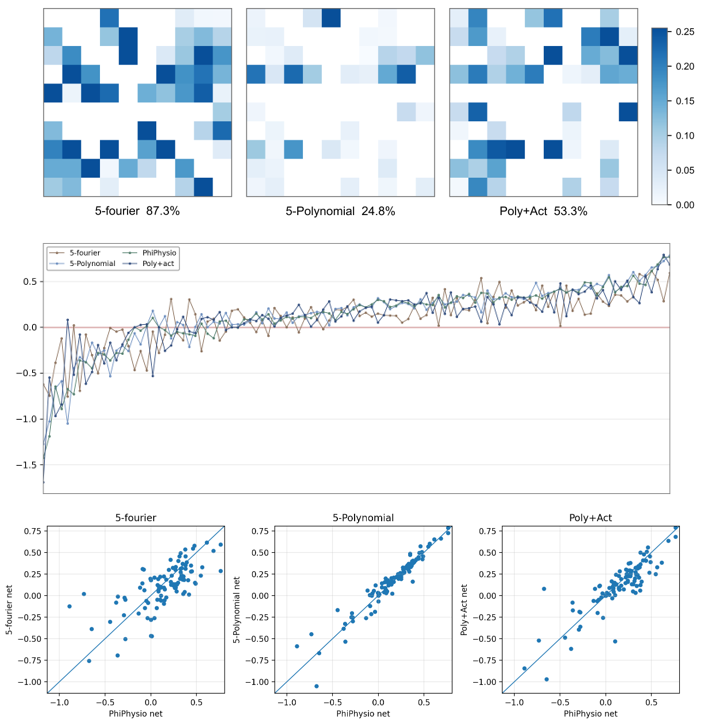}
    \caption{\textbf{Cross-dictionary variability of learned interaction matrices and stability of SOZ-related net-flow summaries (illustrative subject HUP185).}
\textbf{Top row:} absolute difference heat maps between the interaction matrix learned under \texttt{PhiPhysio} and those learned under the other three dictionary configurations (\texttt{5-Fourier}, \texttt{5-Polynomial}, and \texttt{Poly+Act}), shown on a representative 10-node submatrix for visualization. The heat maps are shown only on a representative 10-node submatrix for visual clarity, whereas the percentage reported below each panel is the relative Frobenius discrepancy of the full learned interaction matrix with respect to the \texttt{PhiPhysio} result.
\textbf{Middle row:} per-node NA net flow under the four dictionary configurations.
\textbf{Bottom row:} pairwise comparisons of per-node net flow between \texttt{PhiPhysio} and the other three dictionary configurations. Although the learned interaction matrices differ at the edge-weight level, the corresponding net-flow summaries remain much more consistent across dictionaries.}
    \label{fig:figure8}
\end{figure}

\clearpage

\section{Discussion}

The present framework is designed to balance interpretability, structural regularity, and computational tractability in the joint recovery of graph structure and governing dynamics from trajectory data. In particular, the use of structured basis expansions allows prior information on the local and interaction terms to be incorporated explicitly, while the alternating least-squares formulation keeps the resulting inverse problem computationally manageable. From this perspective, the proposed approach should be viewed as a controlled and interpretable inverse-learning framework rather than as a maximally expressive black-box model \cite{qin2019data,chen2021generalized}.

A brief remark is in order regarding the numerical focus on the scalar-state setting. At the level of both modeling and theory, the present framework is formulated for general states in $\mathbb{R}^d$. The current numerical study is nevertheless restricted to the case $d=1$ for two main reasons. First, this reflects a practical limitation of basis-expansion-based modeling: once the local and interaction terms depend genuinely on multivariate inputs, constructing basis representations that remain accurate, interpretable, and computationally stable becomes substantially more difficult. Second, the scalar setting is not merely a simplified toy case, since many interacting-particle models and learning formulations are themselves scalar, or reduce higher-dimensional interactions to one-dimensional response variables such as distances or norms \cite{lu2021learning,maggioni2021learning}. From this perspective, the present numerical design provides a controlled setting for examining the core inverse-learning problem while remaining consistent with an important class of relevant models.

\section{Future work}

One natural direction is to extend the present numerical study from the scalar-state setting to genuinely higher-dimensional states and to more adaptive function representations. Although the current framework already applies at the level of general $\mathbb{R}^d$-valued states, practical performance in higher dimensions will depend strongly on how the local and interaction functions are represented. A promising direction is therefore to combine the present graph-learning formulation with adaptive structured representations, such as sparsity-promoting basis selection in the spirit of sparse system identification \cite{brunton2016discovering}, low-complexity multivariate dictionaries, or partially learnable basis parameterizations motivated by KAN-like constructions \cite{liu2024kan}. Such extensions could improve expressive power while preserving a degree of interpretability and numerical control.

Another important direction is to relax the shared-mechanism assumption across nodes. In the current IALS formulation, the nodewise local dynamics are assumed to share a common functional template up to node-dependent scaling, and the interaction term is also represented through a shared basis family. While this setting is appropriate for many homogeneous or weakly heterogeneous systems and is useful for preserving identifiability and tractability, more general systems may exhibit genuine node-level heterogeneity in both local responses and interaction laws \cite{fiedler2025recent,patow2024whole,deco2021dynamical,perl2023impact}. Extending the present framework to such settings may require structured parameterizations based on low-rank decompositions, type-based representations, or shared-template-plus-deviation models \cite{argyriou2008convex}. Developing heterogeneous extensions that remain identifiable, stable, and computationally feasible is therefore an important direction for future research.

\section{Conclusion}

In this work, we studied the joint recovery of interaction structure and governing dynamics in stochastic interacting systems from discrete-time trajectory data. By representing both the local drift and the interaction kernel through structured basis expansions, we formulated the estimation problem as a coupled inverse problem for the graph and the dynamical coefficients. Within this framework, we developed two alternating least-squares-type estimators, namely TALS and the diagonal-augmented IALS scheme, the latter being particularly suitable for settings in which nodewise local dynamics share a common functional template up to node-dependent scaling.

At the theoretical level, we established an identifiability result for the diagonal-augmented formulation under a rank-2 joint coercivity condition together with an appropriate normalization convention. At the numerical level, synthetic experiments demonstrated accurate recovery of the interaction graph and dynamical components, favorable convergence with increasing sample size, and robustness to stochastic forcing, observation noise, and basis mismatch. A real-data application to ictal SEEG recordings further illustrated the practical relevance of the framework, showing stable learned dynamical summaries across multiple dictionary choices together with good one-step predictive consistency.

Overall, the proposed framework provides an interpretable, data-adaptive, and computationally tractable approach to learning stochastic interacting systems from trajectory observations. More broadly, it preserves a clear connection between functional parameterization, graph structure, and downstream scientific interpretation, and provides a basis for further extensions to higher-dimensional states and more heterogeneous dynamical mechanisms.

\clearpage
\bibliographystyle{elsarticle-num}
\bibliography{refs}

\clearpage
\appendix

\section{Proof of Theorem~\ref{thm:ials_identifiability}}
\label{app:proof_ials_identifiability}

\begin{proof}[Proof of Theorem~\ref{thm:ials_identifiability}]
Let $(\widetilde A,H,\Phi)\in \widetilde{\mathcal M}\times \mathcal H_{\mathrm{loc}}\times \mathcal H_{\mathrm{int}}$ satisfy
\[
\mathcal E_{L,\infty}(\widetilde A,H,\Phi)=0.
\]
For each $i\in[N]$, define
\[
\mathcal E_{L,\infty}^{(i)}(\widetilde A,H,\Phi)
:=
\frac{1}{L}\sum_{\ell=0}^{L-1}
\mathbb E\!\left[
\left\|
\widetilde a_{ii} H(X_{t_\ell}^i)
+\sum_{j\neq i}\widetilde a_{ij}\Phi(X_{t_\ell}^j)
-\widetilde a_{ii}^* H_*(X_{t_\ell}^i)
-\sum_{j\neq i}\widetilde a_{ij}^* \Phi_*(X_{t_\ell}^j)
\right\|_2^2
\right].
\]
Since $\mathcal E_{L,\infty}=\sum_{i=1}^N \mathcal E_{L,\infty}^{(i)}$, each $\mathcal E_{L,\infty}^{(i)}$ vanishes.

Write
\[
H=q_HH_*+H^\perp,
\qquad
\Phi=q_\Phi\Phi_*+\Phi^\perp,
\]
where
\[
q_H:=\frac{\langle H,H_*\rangle_{L^2(\rho_{\mathrm{loc}})}}{\|H_*\|_{L^2(\rho_{\mathrm{loc}})}^2},
\qquad
q_\Phi:=\frac{\langle \Phi,\Phi_*\rangle_{L^2(\rho_{\mathrm{int}})}}{\|\Phi_*\|_{L^2(\rho_{\mathrm{int}})}^2},
\]
and
\[
\langle H^\perp,H_*\rangle_{L^2(\rho_{\mathrm{loc}})}=0,
\qquad
\langle \Phi^\perp,\Phi_*\rangle_{L^2(\rho_{\mathrm{int}})}=0.
\]
Substituting these decompositions into $\mathcal E_{L,\infty}^{(i)}$ and applying \eqref{eq:rank2Condition} with
\[
(H_1,\Phi_1)=(H_*,\Phi_*),
\qquad
(H_2,\Phi_2)=(H^\perp,\Phi^\perp),
\]
yields
\begin{equation}
\label{eq:coercive_lower_bound_nodewise_short_appendix}
\begin{aligned}
\mathcal E_{L,\infty}^{(i)}(\widetilde A,H,\Phi)\ge c_{\mathcal H}\Big(
&|q_H\widetilde a_{ii}-\widetilde a_{ii}^*|^2\|H_*\|_{L^2(\rho_{\mathrm{loc}})}^2
+\|q_\Phi\widetilde a_{i,-i}-\widetilde a_{i,-i}^*\|_2^2\|\Phi_*\|_{L^2(\rho_{\mathrm{int}})}^2
\\
&+|\widetilde a_{ii}|^2\|H^\perp\|_{L^2(\rho_{\mathrm{loc}})}^2
+\|\widetilde a_{i,-i}\|_2^2\|\Phi^\perp\|_{L^2(\rho_{\mathrm{int}})}^2
\Big).
\end{aligned}
\end{equation}
Since the left-hand side is zero, all four terms on the right-hand side vanish. Hence
\[
q_H\widetilde a_{ii}=\widetilde a_{ii}^*,
\qquad
q_\Phi\widetilde a_{i,-i}=\widetilde a_{i,-i}^*,
\qquad \forall i\in[N].
\]

Moreover, the assumptions
\[
\widetilde a_{i_{\mathrm{loc}}i_{\mathrm{loc}}}^*>0,
\qquad
\|\widetilde a_{i_{\mathrm{int}},-i_{\mathrm{int}}}^*\|_2>0
\]
imply, by choosing $i=i_{\mathrm{loc}}$ and $i=i_{\mathrm{int}}$ in \eqref{eq:coercive_lower_bound_nodewise_short_appendix}, that
\[
\|H^\perp\|_{L^2(\rho_{\mathrm{loc}})}=0,
\qquad
\|\Phi^\perp\|_{L^2(\rho_{\mathrm{int}})}=0.
\]
Therefore,
\[
H=q_HH_*,
\qquad
\Phi=q_\Phi\Phi_*.
\]

It remains to determine the scaling factors $q_H$ and $q_\Phi$. Since
\[
q_H\widetilde a_{ii}=\widetilde a_{ii}^*,
\qquad \forall i\in[N],
\]
we have
\[
q_H\,d(\widetilde A)=d(\widetilde A_*).
\]
Because $\widetilde a_{i_{\mathrm{loc}}i_{\mathrm{loc}}}^*>0$, the vector $d(\widetilde A_*)$ is nonzero. Since both $\widetilde A$ and $\widetilde A_*$ belong to $\widetilde{\mathcal M}$, it follows that
\[
\|d(\widetilde A)\|_2=\|d(\widetilde A_*)\|_2=1.
\]
Taking norms in $q_H\,d(\widetilde A)=d(\widetilde A_*)$ gives
\[
|q_H|=1.
\]
Because both $d(\widetilde A)$ and $d(\widetilde A_*)$ are nonnegative and nonzero, the identity $q_H\,d(\widetilde A)=d(\widetilde A_*)$ implies $q_H>0$. Hence
\[
q_H=1.
\]

Similarly, for the interaction block,
\[
q_\Phi\widetilde a_{i,-i}=\widetilde a_{i,-i}^*,
\qquad \forall i\in[N].
\]
Choosing $i=i_{\mathrm{int}}$, we have $\widetilde a_{i_{\mathrm{int}},-i_{\mathrm{int}}}^*\neq0$, so both $\widetilde a_{i_{\mathrm{int}},-i_{\mathrm{int}}}$ and $\widetilde a_{i_{\mathrm{int}},-i_{\mathrm{int}}}^*$ are nonzero. By the definition of $\widetilde{\mathcal M}$,
\[
\|\widetilde a_{i_{\mathrm{int}},-i_{\mathrm{int}}}\|_2
=
\|\widetilde a_{i_{\mathrm{int}},-i_{\mathrm{int}}}^*\|_2
=1.
\]
Taking norms in
\[
q_\Phi\widetilde a_{i_{\mathrm{int}},-i_{\mathrm{int}}}
=
\widetilde a_{i_{\mathrm{int}},-i_{\mathrm{int}}}^*
\]
gives
\[
|q_\Phi|=1.
\]
Since both vectors are nonnegative and nonzero, this identity implies $q_\Phi>0$. Hence
\[
q_\Phi=1.
\]

Therefore,
\[
d(\widetilde A)=d(\widetilde A_*),
\qquad
\widetilde a_{i,-i}=\widetilde a_{i,-i}^* \quad \forall i\in[N],
\]
and thus
\[
\widetilde A=\widetilde A_*.
\]
Since $H=q_HH_*$ and $\Phi=q_\Phi\Phi_*$ with $q_H=q_\Phi=1$, we conclude that
\[
H=H_*,
\qquad
\Phi=\Phi_*.
\]
Hence
\[
(\widetilde A,H,\Phi)=(\widetilde A_*,H_*,\Phi_*).
\]
\end{proof}

\section{Dictionary specifications}
\label{app:dictionary_specifications}

For completeness, we list here the explicit basis dictionaries used in the experiments. Throughout, we write
\[
\operatorname{sig}(x):=\frac{1}{1+e^{-x}}
\]
for the logistic sigmoid function.

\paragraph{Polynomial+activation function}
The dictionary is
\[
\mathcal D_{\mathrm{Poly+Act}}
=
\left\{
1,\,
x,\,
x^3,\,
e^{-x^2},\,
\tanh(x)
\right\}.
\]

\paragraph{5-Polynomial}
The dictionary is
\[
\mathcal D_{\mathrm{Poly\text{-}5}}
=
\left\{
x,\,
x^2,\,
x^3,\,
x^4,\,
x^5
\right\}.
\]

\paragraph{5-Fourier}
The dictionary is
\[
\mathcal D_{\mathrm{Fourier\text{-}5}}
=
\left\{
1,\,
\sin(x),\,
\cos(x),\,
\cos(2x),\,
\sin(2x),\,
\sin(3x),\,
\cos(3x),\,
\sin(4x),\,
\sin(5x),\,
\cos(4x),\,
\cos(5x)
\right\}.
\]

\paragraph{LocalSmallStrong}
With the implementation setting $\beta_{\tanh}=1$, $\beta_{\mathrm{sig}}=1$, and $\theta=0$, the dictionary is
\[
\mathcal D_{\mathrm{LocalSmallStrong}}
=
\left\{
1,\,
x,\,
x^3,\,
\tanh(x),\,
\operatorname{sig}(x)
\right\}.
\]

\paragraph{PhiPhysio}
With the implementation setting $\beta_{\tanh}=1$, $\beta_{\mathrm{sig}}=1$, and $\theta=0$, the dictionary is
\[
\mathcal D_{\mathrm{PhiPhysio}}
=
\left\{
x,\,
\tanh(x),\,
\operatorname{sig}(x)-\frac12
\right\}.
\]

\paragraph{tanh-zscore-5x2}
This dictionary consists of a linear term together with ten shifted-and-scaled hyperbolic tangent functions:
\[
\mathcal D_{\mathrm{tanh\text{-}zscore\text{-}5\times 2}}
=
\{x\}
\cup
\left\{
\tanh\!\left(\frac{x-\theta}{s}\right)
:\;
\theta\in\{-2,-1,0,1,2\},\;
s\in\{0.5,1.0\}
\right\}.
\]
In the actual implementation, the basis functions are ordered by first fixing $s=0.5$ and then $s=1.0$; within each scale, the thresholds are ordered as $\theta=-2,-1,0,1,2$.

\section{Definitions of prediction error metrics}
\label{app:metric_definitions}

For one-step prediction, let $x_{i,\mathrm{true}}^{(k)}$ and $x_{i,\mathrm{pred}}^{(k)}$ denote the true and predicted next-step values, respectively, for sample index $k=1,\dots,K$ and node index $i=1,\dots,N$. We define the one-step prediction error by
\[
e_i^{(k)}:=x_{i,\mathrm{true}}^{(k)}-x_{i,\mathrm{pred}}^{(k)}.
\]
In the implementation, the sample index $k$ runs over all available one-step prediction instances obtained from the evaluation trajectories.

We first define the global mean absolute error, root mean squared error, and prediction bias by
\[
\mathrm{MAE}_{\mathrm{global}}
:=
\frac{1}{KN}\sum_{k=1}^K\sum_{i=1}^N |e_i^{(k)}|,
\qquad
\mathrm{RMSE}_{\mathrm{global}}
:=
\left(
\frac{1}{KN}\sum_{k=1}^K\sum_{i=1}^N |e_i^{(k)}|^2
\right)^{1/2},
\]
\[
\mathrm{Bias}_{\mathrm{global}}
:=
\frac{1}{KN}\sum_{k=1}^K\sum_{i=1}^N e_i^{(k)}.
\]

Let
\[
s_X:=\operatorname{std}\!\left(\{x_{i,\mathrm{true}}^{(k)}:\;1\le k\le K,\ 1\le i\le N\}\right)
\]
denote the empirical standard deviation of the pooled one-step targets. The normalized error metrics reported at the cohort level are then defined by
\[
\mathrm{NMAE}_{\mathrm{std}}
:=
\frac{\mathrm{MAE}_{\mathrm{global}}}{s_X+\varepsilon},
\qquad
\mathrm{NRMSE}_{\mathrm{std}}
:=
\frac{\mathrm{RMSE}_{\mathrm{global}}}{s_X+\varepsilon},
\]
where $\varepsilon=10^{-12}$ is a small numerical stabilizer. When reported in tables or figures as percentages, these quantities are multiplied by $100$.

At the sample-and-node level, we define the raw relative mean absolute error by
\[
\mathrm{RMAE}_{\mathrm{raw},i}^{(k)}
:=
\frac{|e_i^{(k)}|}{|x_{i,\mathrm{true}}^{(k)}|+\varepsilon}.
\]
Since this quantity can be artificially inflated when the target is close to zero, we also consider a $\tau$-stabilized version,
\[
\mathrm{RMAE}_{\tau,i}^{(k)}
:=
\frac{|e_i^{(k)}|}{\max\!\bigl(|x_{i,\mathrm{true}}^{(k)}|,\tau\bigr)+\varepsilon},
\]
with
\[
\tau:=10^{-3}(s_X+\varepsilon).
\]
The corresponding node-wise summaries are obtained by averaging over the sample index:
\[
\mathrm{RMAE}_{\mathrm{raw},i}
:=
\frac{1}{K}\sum_{k=1}^K \mathrm{RMAE}_{\mathrm{raw},i}^{(k)},
\qquad
\mathrm{RMAE}_{\tau,i}
:=
\frac{1}{K}\sum_{k=1}^K \mathrm{RMAE}_{\tau,i}^{(k)}.
\]

For comparison, we also compute the symmetric mean absolute percentage error at the sample-and-node level:
\[
\mathrm{sMAPE}_i^{(k)}
:=
\frac{|e_i^{(k)}|}
{\frac{1}{2}\bigl(|x_{i,\mathrm{true}}^{(k)}|+|x_{i,\mathrm{pred}}^{(k)}|\bigr)+\varepsilon}.
\]
Unless otherwise stated, the reported sMAPE values are obtained by averaging this quantity over all samples and nodes, and are shown in percentage form by multiplying by $100$.

In addition, for node-wise diagnostics we report
\[
\mathrm{MAE}_i
:=
\frac{1}{K}\sum_{k=1}^K |e_i^{(k)}|,
\qquad
\mathrm{Bias}_i
:=
\frac{1}{K}\sum_{k=1}^K e_i^{(k)}.
\]

To quantify the prevalence of near-zero targets, we define the near-zero target fraction by
\[
\mathrm{Frac}_{\mathrm{near\text{-}zero}}
:=
\frac{1}{KN}\sum_{k=1}^K\sum_{i=1}^N
\mathbf 1\!\left(|x_{i,\mathrm{true}}^{(k)}|<\tau\right).
\]
This quantity is used only as a diagnostic for interpreting relative errors in regimes where the target magnitude is very small.

\end{document}